\documentclass[11pt]{article}

\usepackage[margin=1in]{geometry}

\usepackage{enumitem}
\usepackage[utf8]{inputenc}
\usepackage{amsmath,amssymb}
\usepackage{color}
\usepackage{calc}
\usepackage{tikz}
\usetikzlibrary{positioning,arrows,decorations.pathreplacing,shapes}
 \usetikzlibrary{calc}
\usepackage{multirow}
\usepackage{boxedminipage}
 \usepackage{xifthen}
 \usepackage{tabularx}

 \usepackage[sort&compress,square,comma,numbers]{natbib}

 \usepackage{setspace}

\usepackage{algorithmic}
\usepackage[linesnumbered,vlined,ruled]{algorithm2e}
\usepackage{hyperref}
\newcommand{\footremember}[2]{%
    \footnote{#2}
    \newcounter{#1}
    \setcounter{#1}{\value{footnote}}%
}

\newcommand{\R}{\mathbf{R}}

\newcommand{\x}{\mathbf{x}}
\newcommand{\y}{\mathbf{y}}

\newcommand{\h}{\mathbf{h}}
\newcommand{\w}{\mathbf{w}}

\DeclareMathOperator{\E}{\mathbb E}

\DeclareMathOperator{\work}{work}

\DeclareMathOperator{\argmin}{argmin}

\newtheorem{claim}{Claim}[section]
\newtheorem{example}{Example}[section]

\newtheorem{proposition}{Proposition}

\newtheorem{theorem}{Theorem}
\newtheorem{corollary}{Corollary}

\newtheorem{lemma}{Lemma}
\newtheorem{remark}{Remark}
\newenvironment{proof}[1][]{\par \noindent {\bf Proof #1}\ }{\hfill$\Box$\par \vspace{11pt}}
\def\qed{\hfill$\Box$}

\newlength{\RoundedBoxWidth}
\newsavebox{\GrayRoundedBox}
\newenvironment{GrayBox}[1]%
   {\setlength{\RoundedBoxWidth}{.93\textwidth}
    \def\boxheading{#1}
    \begin{lrbox}{\GrayRoundedBox}
       \begin{minipage}{\RoundedBoxWidth}}%
   {   \end{minipage}
    \end{lrbox}
    \begin{center}
    \begin{tikzpicture}%
       \node(Text)[draw=black!80,fill=white,rounded corners,%
             inner sep=2ex,text width=\RoundedBoxWidth]%
             {\usebox{\GrayRoundedBox}};
        \coordinate(x) at (current bounding box.north west);
        \node [draw=white,rectangle,inner sep=3pt,anchor=north west,fill=white]
        at ($(x)+(6pt,.75em)$) {\boxheading};
    \end{tikzpicture}
    \end{center}}

\newenvironment{defproblemx}[2][]{\noindent\ignorespaces%
                                \FrameSep=6pt%
                                \parindent=0pt%
                \vspace*{-1.5em}
                \ifthenelse{\isempty{#1}}{%
                  \begin{GrayBox}{\textsc{#2}}%
                }{%
                  \begin{GrayBox}{\textsc{#2} parameterized by~{#1}}%
                }
                \begin{tabular*}{\textwidth}{@{\hspace{.1em}} >{\itshape} p{1.8cm} p{0.8\textwidth} @{}}%
            }{
                \end{tabular*}%
                \end{GrayBox}%
                \ignorespacesafterend
            }

\usepackage{rotating}
\usepackage{framed}

\renewcommand\cite{\citep}
\renewcommand\citet{\citep}
\renewenvironment{proof}[1][{\bfseries Proof}]{{\bfseries #1. }}{\qed}

\title{Dynamic Resource Allocation in the Cloud\\with Near-Optimal Efficiency }

\author{Sebastian Perez-Salazar\footremember{2}{Georgia Institute of Technology, \texttt{sperez@gatech.edu}}
	\and Ishai Menache\footremember{1}{Microsoft Research, \texttt{ishai@microsoft.com}}
	\and Mohit Singh\footremember{3}{Georgia Institute of Technology, \texttt{mohit.singh@isye.gatech.edu}}
	\and Alejandro Toriello\footremember{4}{Georgia Institute of Technology, \texttt{atoriello@isye.gatech.edu}}
}

\begin{document}

\maketitle

\begin{singlespace}
\begin{abstract}
Cloud computing has motivated renewed interest in resource allocation problems with new consumption models. A common goal is to share a resource, such as CPU or I/O bandwidth, among distinct users with different demand patterns as well as different quality of service requirements. To ensure these service requirements, cloud offerings often come with a service level agreement (SLA) between the provider and the users. An SLA specifies the amount of a resource a user is entitled to utilize. In many cloud settings, providers would like to operate resources at high utilization while simultaneously respecting individual SLAs. There is typically a tradeoff between these two objectives; for example, utilization can be increased by shifting away resources from idle users to ``scavenger'' workload, but with the risk of the former then becoming active again. We study this fundamental tradeoff by formulating a resource allocation model that captures basic properties  of cloud computing systems, including SLAs, highly limited feedback about the state of the system, and variable and unpredictable input sequences. Our main result is a simple and practical algorithm that achieves near-optimal performance on the above two objectives. First, we guarantee nearly optimal utilization of the resource even if compared to the omniscient offline dynamic optimum. Second, we simultaneously satisfy all individual SLAs up to a small error. The main algorithmic tool is a multiplicative weight update algorithm, and a primal-dual argument to obtain its guarantees. We also provide numerical validation on real data to demonstrate the performance of our algorithm in practical applications.

\vspace{1em}
\end{abstract}
\end{singlespace}

\section{Introduction}


Cloud computing has motivated renewed interest in resource allocation, manifested in new consumption models (e.g., AWS spot pricing), as well as the design of resource-sharing platforms \citep{hindman2011mesos,vavilapalli2013apache}. These platforms need to support a heterogenous set of users, also called tenants, that share the same physical computing resource, e.g., CPU, memory, I/O bandwidth. Providers such as Amazon, Microsoft and Google offer cloud services with the goal of benefiting from economies of scale. However, the inefficient use of resources -- over-provisioning on the one hand or congestion on the other -- could result in a low return on investment or in loss of customer goodwill, respectively. Hence, resource allocation algorithms are key for efficiently utilizing cloud resources.






To ensure quality of service, cloud offerings often come with a \emph{service level agreement} (SLA) between the provider and the users. An SLA specifies the amount of resource the user is entitled to consume. Perhaps the most common example is renting a virtual machine (VM) that guarantees an explicit amount of CPU, memory, etc. Naturally, VMs that guarantee more resources are more expensive. In this context, a simple allocation policy is to assign each user the resources specified by their SLAs. However, such an allocation can be wasteful, as users may not need the resource at all times. In principle, a dynamic allocation of resources can increase the total efficiency of the system. However, allocating resources dynamically without carefully accounting for SLAs can lead to user dissatisfaction.

Recent scheduling proposals address these challenges through work-maximizing yet fair schedulers \citep{zaharia2010delay,ghodsi2011dominant}. However, such schedulers do not have explicit SLA guarantees. On the other hand, other works focus on enforcing SLAs \citep{rayon,morpheus,grandl2016altruistic}, but do not explicitly optimize the use of extra resources.

Our goal in this work is to understand the fundamental tradeoff between high utilization of resources and SLA satisfaction of individual users. In particular, we design algorithms that guarantee \emph{both} near optimal utilization as well as the satisfaction of individual SLAs, simultaneously. To that end, we formulate a basic model for online dynamic resource allocation. We focus on a single divisible resource, such as CPU or I/O bandwidth, that has to be shared among multiple users. Each user also has an SLA that specifies the fraction of the resource it expects to obtain. The actual demand of the user is in general time-varying, and may exceed the fraction specified in the SLA. As in many real systems, the demand is not known in advance, but rather arrives in an online manner. Arriving demand is either processed or queued up, depending on the resource availability. In many real-world scenarios, it is difficult to measure the actual demand size (see, e.g., \cite{SQLVM}). Accordingly, we assume that the system (and the underlying algorithm) receives only a simple \emph{binary feedback} per user at any given time: whether the user queue is empty (the user's work arriving so far has been completed), or not. This is a plausible assumption in many systems, because one can observe workload activity, yet anticipating how much of the resource a job will require is more difficult. Additionally, it also models settings where demands are not known in advance.

While online dynamic resource allocation problems have been studied in different contexts and communities (see Section \ref{sec:related} for an overview), our work aims to address the novel aspects arising in the cloud computing paradigm, particularly the presence of SLAs, the highly limited feedback about the state of the system, and a desired robustness over arbitrary input sequences. For the algorithm design itself, we pay close attention to practicality; our approach involves fairly simple computations that can be implemented with minimal overhead of space or time. Our algorithm achieves nearly optimal utilization of the resource, as well as approximately satisfying the SLA of each individual user. We see two main use-cases for the algorithm:
\begin{itemize}
\item In enterprise settings (``private cloud''), different applications or organizations share the same infrastructure. These often have SLAs, but providers would still like to maximize the ROI by maximizing utilization \citep{mercury}.
\item In public clouds, users buy VMs, which are offered at different ``sizes'' (which is practically the SLA). In addition, the service providers offer ``best-effort'' alternatives, such as Azure Batch (MS) or Spot instances (AWS). In our model, these services can be modeled by giving an SLA of zero. Here, satisfying the VM SLAs and achieving high utilization are both important; indeed, the provider is paid for the best-effort workloads only if it completes these jobs. Our work can be viewed as a principled way to accommodate such services, and even give VMs better service than expected, an important consideration as public cloud offerings gradually become commoditized.
\end{itemize}

\subsection{The Model}

We consider the problem of having multiple tenants or users sharing a single resource, such as CPU, I/O or networking bandwidth. For simplicity, we assume that the total resource capacity is normalized to $1$. We have $N$ users sharing the resource, a finite but possibly unknown discrete time horizon indexed $t=1,\ldots, T$, and an underlying queuing system. For each user $i$, we are also given an expected share of resource $\beta(i)\geq 0$ satisfying $\sum_{i=1}^N \beta(i)\leq 1$. The input is an online sequence of workloads $L_1,\ldots,L_T\in \R_+^N$, where $L_t(i)\geq 0$ corresponds to $i$'s workload  arising at time $t$. The system maintains a queue $Q_t(i)$, denoting $i$'s remaining work at time $t$. In our model, the decision maker does \emph{not} have direct access to the values of the queues or the workloads. This allows us to consider settings where the job sizes are not known in advance and minimal information is available about the underlying system, a regular occurrence in many cloud settings. At time $t$, the following happens:
\begin{enumerate}
\item \textbf{Feedback:} The decision maker observes which queues are non-empty (the set of users $i$ with $Q_t(i)>0$, the \emph{active} users), and which are empty ($ Q_t(i) = 0 $, the \emph{inactive} users). 
	
\item \textbf{Decision:} The decision maker updates user resource allocations $h_{t}(i)$, satisfying $\sum_i h_t(i)\leq 1$.
	
\item \textbf{Update:} The load $L_t(i)$ for each $i$ arrives and each user processes as much of the work from the queue plus the arriving workload as possible. The work completed by user $i$ in step $t$ is
	$$w_t(i):=\min\{h_t(i), L_t(i)+Q_t(i)\}.$$
	The queues at the end of the time step are updated accordingly,
	\[Q_{t+1}(i) = \max\{0, L_t(i) + Q_t(i) - h_t(i)\}.\]
\end{enumerate}
We assess the performance of any algorithm based on two measures.
\begin{enumerate}
	\item \textbf{Work Maximization.} The algorithm should maximize the total work completed over all users, and thus utilize the resource as much as possible.
	\item \textbf{SLA Satisfaction.} The algorithm should (approximately) satisfy the SLAs in the following manner. The work completed by user $i$ up to any time $1\leq t \leq T$ should be no less than the work completed for this user up to $t$ if it were given a constant $ \beta(i) $ fraction  of the resource over the whole horizon.
\end{enumerate}
Achieving either of the criteria on their own is straightforward. A greedy strategy that takes away resources from an idle user and gives them to any user whose queue is non-empty is approximately work-maximizing (see Appendix~\ref{sec:localgreedy} for details). On the other hand, to satisfy the SLAs, we give each user a \emph{static} assignment of $h_t(i):=\beta(i)$ for all $t$. Naturally, the two criteria compete with each other; the following examples illustrate why these simple algorithms do not satisfy both simultaneously.

\begin{figure}[h!]
	
	\centering
	
	\includegraphics[width=1.01\linewidth]{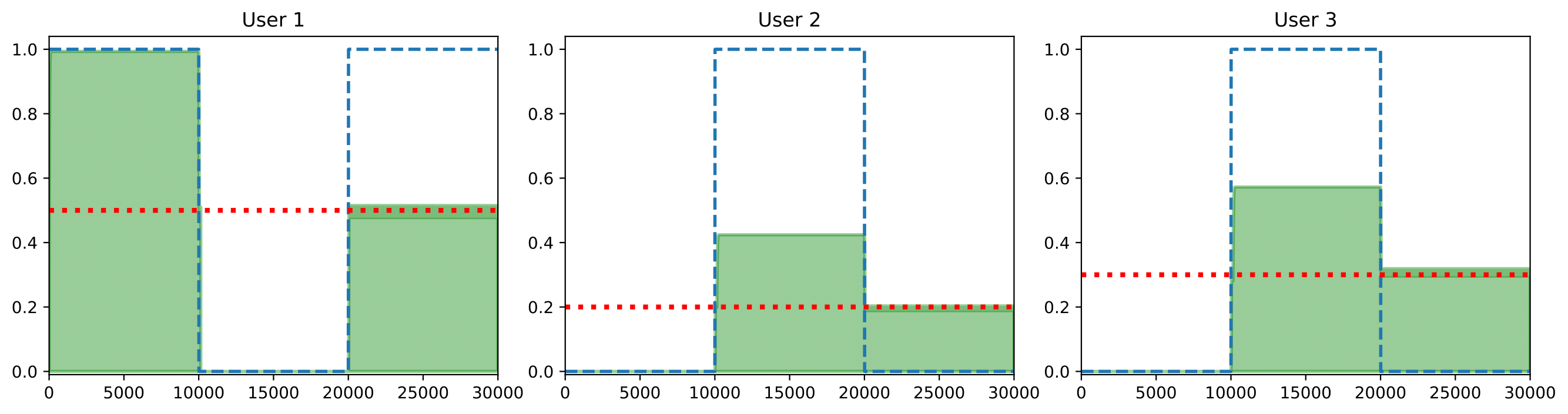}
	\caption{Example of loads and work for 3 users. The blue dashed lines show each user's workload. The green areas represent the work done by the users. The red dashed lines depict SLAs.}
	\label{fig:graph}
\end{figure}

\begin{example}
	We have a shared system with three users and corresponding SLAs $\beta(1)= 0.5, \beta(2)=0.2$ and $\beta(3)=0.3$. Loads are defined by
	{\small\[
		L_t(1)=\begin{cases}
		1 & t=1,\ldots,T/3,2T/3+1,\ldots,T\\
		0 & t=T/3+1,\ldots,2T/3
		\end{cases}, \text{ and} \quad L_t(2)=L_t(3)= 1- L_t(1).
		\]}
	We assume that $T$ is divisible by $3$. In Figure~\ref{fig:graph} we show the three users' loads in blue dashed lines and the corresponding SLAs in dotted red lines. The static solution given by the SLAs, i.e.\ $h_t(i)=\beta(i)$ for all $t$, ensures a total of $5T/6$ work done. However, the dynamic policy shown in the green line, given by
	{\small
		\begin{center}
			\begin{tabular}{|c|c|c|c|}
				\hline
				$t$ & $[0,T/3]$ & $[T/3+1,2T/3]$ & $[2T/3+1,T]$\\
				\hline
				$h_t(1)$ & $1$    & $0$  & $0.5$\\
				\hline
				$h_t(2)$ & $0$    & $0.4$  & $0.2$\\
				\hline
				$h_t(3)$ & $0$    & $0.6$  & $0.3$\\
				\hline
			\end{tabular}
		\end{center}
	}
\noindent ensures $T$ work is done (the green area). Moreover, it also ensures SLA satisfaction at all times. An alternative policy is
	{\small
		\begin{center}
			\begin{tabular}{|c|c|c|c|}
				\hline
				$t$ & $[0,T/3]$ & $[T/3+1,2T/3]$ & $[2T/3+1,T]$\\
				\hline
				$h_t(1)$ & $1$    & $0$  & $0$\\
				\hline
				$h_t(2)$ & $0$    & $1$  & $0$\\
				\hline
				$h_t(3)$ & $0$    & $0$  & $1$\\
				\hline
			\end{tabular}\\
		\end{center}
	}
\noindent which is also work maximizing. However, it does not ensure SLA satisfaction. Indeed, this policy does not satisfy user $3$'s SLA at any time in $(T/3,2T/3]$.
\end{example}

We remark that achieving both criteria is relatively simple if we allow the decision maker to observe demand or even the queue length. This can be achieved by first allocating to each user as much of the resource as necessary up to their SLA, and then distributing the remaining resource arbitrarily among users with additional demand. 
The versatility of our setting stems from the limited feedback in the form of binary information about idle and busy users. In our cloud computing context, full demand information or even visible queue lengths are unrealistic assumptions.

\subsection{Our Results and Contributions}

We design a simple and efficient online algorithm that achieves approximate work maximization as well as approximate SLA satisfaction even in the limited feedback model that we consider. For work maximization, we analyze the performance by comparing our algorithm to the \emph{optimal offline dynamic allocation} that knows all the data up front. In contrast, our online algorithm receives limited  feedback even in an online setting.  Thus, our aim  is to minimize the quantity $$\work_{\h_1^*,\ldots,\h_T^*} -\work_{\text{alg}} = \sum_{t=1}^T \sum_i w_t^*(i)  - \sum_{t=1}^T \sum_i w_t(i),$$
where $\work_{{\h}_1^*,\ldots,{\h}_T^*}$ is the optimal offline work done by dynamic allocations $\h_1^*,\ldots,\h_T^*$ , $\w_t^*= (w_t^*(1),\ldots, w_t^*(N)) $ is the work done at time $t$ by these allocations, and $\work_{\text{alg}}$ is the work done by the algorithm with allocations $\h_1,\ldots,\h_T$ and work $\w_t= (w_t(1),\ldots, w_t(N))$ at time $t$. The objective of the decision maker is to \emph{minimize} this quantity by constructing a sequence of good allocations that approach the best allocations in hindsight. Note that our benchmark is \emph{dynamic}, rather than the more common \emph{static} offline optimum usually considered in regret minimization \citep{arora2012multiplicative,hazan2016introduction,shalev2012online}. Similarly, for SLA satisfaction, our benchmark is the total work done for a user if they were given $\beta(i)$ resources for each time $1\leq t\leq T$. We give a bi-criteria online algorithm that achieves nearly the same performance as the benchmarks if the resources for the latter are slightly more constrained than that of the algorithm.
Algorithm~\ref{alg:SMB1}, which we formally describe in Section \ref{sec:algorithm}, follows a multiplicative weight approach. The idea is to boost the allocations of active users by a factor greater than $1$, with more emphasis on users with current allocation below their SLA. This intuition translates into a simple update that ensures high utilization of the resource and SLA satisfaction; formally:

\begin{theorem}\label{thm:main} For any input parameter $0<\varepsilon\leq \frac{1}{10}$, SLAs $\beta=(\beta(1),\ldots,\beta(N))$ satisfying $\beta(i) \geq  2\frac{\varepsilon}{N}$, and online loads $L_1,\ldots,L_T\in \R_{\geq 0}^N$, Algorithm~\ref{alg:SMB1} achieves the following guarantees:
	\begin{enumerate}
		\item \textbf{Approximate Work Maximization.}  Let $\h_1^*,\ldots,\h_T^*\in [0,1]^N$ be an optimal offline sequence of allocations such that $\sum_i h_t^*(i)=1$ for all $1\leq t\leq T$. Then
		\[
		\textstyle	\work_{\text{alg}} \geq  (1-\varepsilon) \work_{\h_1^*,\ldots,\h_T^*} -  \mathcal{O} \left({N {\varepsilon^{-2}\log(N/\varepsilon) }}\right).
		\]

		\item \textbf{Approximate SLA Satisfaction.} There exists $\tilde{p} = \tilde{p}(N,\varepsilon) =\mathcal{O} \left( N^2 \varepsilon^{-3}  \log (N/\varepsilon)  \right)$ such that for any user $i$ and time $t$, if we take $\h_1',\ldots,\h_T'\in [0,1]^N$ to be any sequence of allocations with $h_t'(i)\leq \beta(i)$, then
		\[
			\sum_{\tau=1}^{t} w_\tau(i) \geq \left(1-2\varepsilon\right) \sum_{\tau = 1 }^{t}  w_\tau'(i) -\beta(i) \tilde{p} ,
		\]
		where $\w_t'$ is the work performed by the allocations $\h_1',\ldots,\h_T'$.
	\end{enumerate}
\end{theorem}

The first part of Theorem~\ref{thm:main} asserts that if $T$ is known for the system, we can achieve $\work_{\text{offline}}-\work_{\text{alg}}\leq \mathcal{O}(T^{2/3}\log T)$ with $\varepsilon= \Theta\left(\frac{N^{1/3}}{T^{1/3}}\right)$. 
However, this choice of $\varepsilon$ is suboptimal for SLA satisfaction. We can achieve an improved bound for SLA satisfaction (when $t=T$) by picking $\varepsilon= \Theta\left(\frac{N^{1/3}}{T^{1/4}}\right)$. If $T$ is unknown, we can use a standard doubling trick, see for example \citep{shalev2012online}. Summarizing, we obtain the following result.
\begin{corollary}\label{cor:UpperboundAlg}
	For $\varepsilon= \Theta\left( \frac{N^{1/3}}{T^{1/4} }  \right)$, Algorithm~\ref{alg:SMB1}  guarantees
	$$\textstyle \work_{\h_1^*,\ldots,\h_T^*} - \work_{\text{alg}} \leq \mathcal{O}(N^{1/3} T^{3/4} + N^{1/3} T^{1/2} \log(N T) ) =\mathcal{O}(N^{1/3}T^{3/4}) ,$$
	where $\h_1^*,\ldots,\h_T^*$ are optimal offline dynamic allocations.
\end{corollary}	

As $T$ grows, Corollary~\ref{cor:UpperboundAlg} guarantees that the \emph{rate} of work done by our algorithm $\frac{\work_{\text{alg}}}{T}$,
approaches the rate of work done by the optimal dynamic solution $\frac{\work_{{\h}_1^*,\ldots,{\h}_T^*}}{T}$.
We emphasize again that this is a stronger guarantee using the much more powerful optimal dynamic solution as a benchmark, rather than the typical static allocation used in regret analysis for online algorithms. Such a guarantee can be obtained in our model because the incomplete work remains stored in the queues until we are able to finish it; this allows the algorithm to \emph{catch up} with the incomplete work. 

The second result in Theorem~\ref{thm:main} states that the work done by any individual user is comparable to the work done by their promised SLA. In other words, the user's queue length is not much larger than it would be under a static SLA allocation. By using $\varepsilon= \Theta\left( \frac{N^{1/3}}{T^{1/4}} \right)$ we obtain the following result.
\begin{corollary}\label{cor:Queues}
	Let $\varepsilon= \Theta\left( \frac{N^{1/3}}{T^{1/4}}  \right) $. For a user $i$,  $Q_T(i) \leq Q_{T}'(i) + \mathcal{O}(NT^{3/4}\log (NT))$, where $Q_t$ is the queue given by Algorithm~\ref{alg:SMB1} and $Q_t'$ is the queue induced by any dynamic policy $\h_1',\ldots,\h_T'$ with $h_t'(i)\leq \beta(i)$.
\end{corollary}

We remark that we need to bound the SLAs away from $\frac{\varepsilon}{N}$; in Theorem~\ref{thm:main} we use the bound $\beta(i)\geq 2\frac{\varepsilon}{N}$. This condition is necessary in our analysis to guarantee that there is an over-provisioned ($h_t(i)>\beta(i)$) user from which we can move allocation to an under-provisioned user. See Theorem~\ref{thm:main2} for details and a more relaxed bound.

%

Corollary~\ref{cor:UpperboundAlg}'s guarantee is near-optimal asymptotically in $T$ in terms of work maximization, as the following result shows. The proof of this result appears in Appendix~\ref{sec:Lower}.
\begin{theorem}\label{thm:GeneralLB_main}
	For any online deterministic algorithm $\mathcal{A}$ for our model, there is a sequence of online loads $L_1,\ldots,L_T$ such that $\work_{\h_1^*,\ldots,\h_T^*} - \work_\mathcal{A} = \Omega\left(\sqrt{T}\right)$, where $\h_1^*,\ldots,\h_T^*$ are optimal offline dynamic allocations.
\end{theorem}

Our algorithm follows a mirror descent approach \citep{ben2001lectures,hazan2016introduction}. Unable to see the lengths of the queues, a (naive) approach is to pretend that active users have gigantic queues. From this approach we extract a simple update rule that multiplicatively boosts active users; however, active users who are under their SLA are boosted slightly more than other active users. If inactive users are assigned more than $\varepsilon$ of the resource, where $\varepsilon$ is the algorithm parameter, active users ramp up their allocation in few iterations. In the opposite case, at least $1-\varepsilon$ of the resource is assigned to active users, and the slight boost to users below their SLA ensures a healthy re-balancing of the resource.  We show that this efficient heuristic strategy is enough to achieve approximate work maximization and SLA satisfaction. We remark that the mirror descent interpretation is only used to provide intuition for the algorithm, and our proofs follow a different path than the usual mirror descent analysis. Later, we detail a slight modification that enjoys the same theoretical guarantees as Algorithm~\ref{alg:SMB1} but in practice exhibits more desirable behavior (see Algorithm~\ref{alg:SMB2}). Intuitively, among active users, the modified algorithm tries to keep allocations proportional to their SLAs. This behavior is appealing; for example, a user $A$ with twice the SLA of another user $B$ would expect in practice to perform at least twice as much work. Similarly, user $B$ would expect to receive no less than half of user $A$'s allocation. This second algorithm exhibits another interesting feature; it can be applied to over-committed regimes with $\sum_{i=1}^N \beta (i) > 1$, remain work-maximizing, and satisfy a normalized version of SLA satisfaction (see Section~\ref{sec:Extension}).


The analysis of the algorithm relies on a primal-dual fitting approach. For work maximization, we can write the offline dynamic optimal allocation as a solution to a linear program and then construct feasible dual solutions with objective value close to the algorithm's resource utilization. A crucial ingredient of the algorithm is the use of entropic projection on the \emph{truncated} simplex, which ensures every user gets at least a $\varepsilon/N$ fraction of the resource at all times. Intuitively, this means any user with a non-empty queue will recover their SLA requirement in a few steps.

We do an extensive analysis of our algorithm on synthetic data as well as real data obtained from CPU traces of a production service in Microsoft's cloud. We aim to quantify the performance of the algorithm on three objectives, (i) work maximization, (ii) SLA guarantee and (iii) queue behaviour. While our theoretical results give guarantees for these objectives, we show experimentally that the algorithm exceeds these guarantees. We benchmark the algorithm against natural online algorithms, such as a static allocation as given by the SLA guarantee, or the algorithm that aims to proportionally distribute the resource among active clients (according to their SLA). We also benchmark against offline algorithms that know all input data up front; our algorithm's performance on various measures is comparable to the offline algorithms.


This work is organized as follows. In Section \ref{sec:algorithm}, we present the preliminaries and the basic version of the multiplicative weight algorithm. Section \ref{sec:Analysis} contains the proof of Theorem \ref{thm:main} in the bi-criteria form. We split the proof into two parts: work maximization in Section \ref{sec:throughput} and SLA satisfaction in \ref{sec:SLAsat}. In Section \ref{sec:Extension}, we present the extension of our algorithm and its guarantees. Finally, in Section~\ref{sec:Numerical} we present numerical experiments that empirically validate our results, and conclude in Section \ref{sec:Conclusions}.


\subsection{Related Work}\label{sec:related}


There has been growing interest in resource allocation problems arising from cloud computing applications, both from a practical as well as a theoretical standpoint \citep{grandl2016altruistic,hindman2011mesos,mercury,morpheus,rayon,SQLVM,vavilapalli2013apache,narasayya2015sharing,menache2015online}. The focus of many of these works has been to understand the trade-offs between efficiency and ensuring guarantees to individual users.

On the more theoretical side, the cloud computing allocation problem has been modeled as a stochastic allocation problem \citep{maguluri2014heavy, maguluri2012stochastic, maguluri2014scheduling}. The underlying models draw inspiration from a large body of work on stochastic network control, originating from the seminal work of
\citet{tassiulas1990stability,tassiulas1993dynamic}, followed by additional related research, e.g., on uplink and downstream scheduling in wireless networks \citep{neely2007optimal,neely2008order}. The analytical results in these papers characterize the stability region of the arrival processes under certain stochastic assumptions (e.g., i.i.d.~processes), and suggest algorithms that achieve maximal throughput. The main distinction between these works and ours is that we assume an \emph{adversarial} input, i.e., we do not make stationary distributional assumptions on the input. Another difference is that our model centers on the notion of an SLA which is known to the algorithm. This allows us to address the over-commited case (see Algorithm~\ref{alg:SMB2}), which is especially relevant in cloud settings.


Despite these modeling differences, there are some parallels worth mentioning. For example, we design algorithms with rate of work $\frac{\work_{\text{alg} } }{T}$ approaching the optimal offline rate of work; this translates to the average delay $\frac{1}{T}\sum_{t=1}^T \sum_{i=1}^N Q_t(i)$ converging to the optimal offline average delay. This result can be compared to the limiting behavior of the Markov process in many of the stochastic network models. For example, \cite{neely2008order} studies a system of $N$ users connecting to a server via ON/OFF channels. The paper presents an algorithm that ensures bounded average delay $\limsup_T \frac{1}{T} \E\left[ \sum_{t=1}^T \sum_{i=1}^N Q_{i}(t)  \right]$ for any input within the interior of the stability region; here $Q_i(t)$ is the $i$-th user's backlog (queue length) at time $t$. Much of the stochastic network literature assumes full knowledge of queue lengths, e.g.\ the LCQ policy in \citet{tassiulas1993dynamic}, although there are studies that limit the information available to the decision maker in a similar fashion to our model (see \citet{li2009energy, neely2007optimal,maguluri2014scheduling,shirani2010mimo}). 


More broadly, the general problem of online resource allocation has been studied in both stochastic and adversarial settings; we refer the reader to the books \citep{albers2003online,borodin2005online,srikant2013communication} on the topic. Our work differs from the aforementioned lines of research by combining the three following elements already present in the literature. First, as mentioned above, we digress from the stochastic arrival model to the adversarial setting and worst-case analysis. This makes our algorithm robust against unpredictable users' demands. For instance, demands in the morning could be totally different from demands in the afternoon or the morning of the next day. We are able to provide a single strategy that adapts easily to any scenario. Second, we provide simple online algorithms that perform well even under \emph{limited feedback}, a typical situation in cloud systems in which we can only determine the utilization of a resource after it has been allocated. Finally, we consider SLA satisfaction as a measure of user contentment, and seek to satisfy it up to a small error.


There is now an extensive literature devoted to the pricing of cloud computing services. In \citep{macias2011genetic} the authors study a genetic model for generating a suitable pricing function in the cloud market. In \citep{gera2011learning}, pricing is studied via a revenue management formulation to address resource provisioning decisions. See also \citep{passacantando2016service,sharma2012pricing} for more pricing models. A closely related topic is fairness in resource allocation \citep{zaharia2010delay,ghodsi2011dominant}. Although we do not directly consider pricing, work maximization could be interpreted as a way to obtain extra revenue by allocating unused resources to active users.

More recently, there has been work considering over-commitment in the cloud \citep{cohen2019overcommitment,gordon2011ginkgo,dabbagh2015efficient}, that is, selling resources beyond server capacity. One of the objectives of over-commitment is to reduce the number of servers opened in order to minimize energy consumption. In our basic model, we do not assume over-commitment, yet our algorithm can still be applied to that setting (see Algorithm~\ref{alg:SMB2}). Specifically, we obtain a normalized version of SLA satisfaction in the over-commitment setting, where the guarantees depend on how much the system is over-committed (see Section~\ref{sec:Extension}).

A fundamental tool in our design is \emph{mirror descent} algorithms \citep{ben2001lectures}. These first-order iterative algorithms have been widely used in optimization \citep{ben2001lectures}, online optimization and machine learning \citep{hazan2016introduction,mohri2018foundations,shalev2012online} to generate update policies under limited feedback. Similarly, multiplicative weights algorithms have been widely studied in optimization \citep{plotkin1995fast,arora2012multiplicative}, online convex optimization \citep{hazan2016introduction}, online competitive analysis \citep{buchbinder2009design} and learning theory \citep{freund1997decision,shalev2012online}. Our results bear some resemblance to regret analysis, where typically the benchmark is the optimal offline static policy \citep{shalev2012online,hazan2016introduction,abernethy2008optimal,bubeck2012regret,freund1997decision}; the use of a dynamic benchmark (as in our work) is scarcer in the literature, see e.g.\ \citet{hall2015online,mokhtari2016online,zinkevich2003online,zhang2017improved}.

\section{Algorithm}\label{sec:algorithm}


\subsection{Preliminaries}\label{sec:prelim}

For $N\geq 1$, we identify the set of \emph{users} with the set $[N]=\{1, \ldots, N \}$. For $0< \varepsilon< 1$, we call an allocation $\h =(h(1),\ldots, h(N) )\in [0,1]^N$ a \emph{$(1-\varepsilon)$-allocation} if $\sum_i h(i)\leq 1-\varepsilon$. We assume $N\geq 2$, that is, the system consists of at least two users.

For any $t$, we define the set of \emph{active users} at that time as the set of users with non-empty queue, and denote this set by $A_t$. Observe that $h_t(i) = w_t(i)$ for all active users. Let $B_t=[N]\setminus A_t$ be the sets of users with empty queues at time $t$; we call these users  \emph{inactive}. $A_t$ and $B_t$ correspond to the feedback given to the decision maker. Also, let $A_t^1 = \{ i\in A_t : h_t(i)< \beta(i)  \}$ be the set of active users with allocation below their SLA and $A_t^2 = A_t\setminus A_t^1$ be the set of active users receiving at least their SLAs.

We assume without loss of generality that the allocations set by the decision maker always add up to $1$. We propose an algorithm that uses a multiplicative weight strategy to boost a subset of users by multiplying their allocation by factor greater than one. Because the allocations do not sum to one after applying the update, we then project them onto the simplex using the KL-divergence metric. Furthermore, to ensure no user gets an allocation arbitrarily close to zero, we in fact project onto the \emph{truncated simplex},
\[
\Delta_\varepsilon = \{ \x =(x(1),\ldots,x(N)) :  \|\x\|_1 = 1, x(i)\geq \varepsilon/N, \forall i  \}.
\]
To fix notation, let $\pi_{\Delta_\varepsilon}(\cdot)$ be the \emph{projection function} onto $\Delta_\varepsilon$ using Kullback–Leibler divergence (KL-divergence for short), i.e., $\pi_{\Delta_\varepsilon}(\y) := \argmin_{\x \in \Delta_\varepsilon} \sum_i x(i)\log ( x(i) / y(i) )$, where $\y=(y(1),\ldots,y(N) )\in \R_{\geq 0}^N$. In Appendix \ref{sec:projection}, we show how to efficiently compute this projection. The following proposition states some basic facts that are useful in our analysis. The proof appears in Appendix \ref{sec:projection}.
\begin{proposition}\label{prop:projection}
	Let $\y \in \R_+^N$, $\x =\pi_{\Delta_\varepsilon}(\y)$, and $S=\{ i : x(i) = \frac{\varepsilon}{N}  \}$. Then:
	\begin{enumerate}[label=(\alph*)]
		\item If $y(1)\leq y(2)\leq \cdots \leq y(N)$, then $S=\{1,\ldots, k  \}$ for some $k\geq 0$.
		
		\item $x(i) = y(i) e^{\mu_i} C$, where   $ C= \left( \frac{1-\frac{\varepsilon}{N}|S|}{\sum_{j \not\in S} y(j) }\right)$, $\mu_i \geq 0$ for all $i$ and $\mu_i=0$ for $i\notin S$.
		
		\item $\x$ can be computed in $\mathcal{O}(N\log N)$ time.
		
	\end{enumerate}
\end{proposition}

\subsection{The Multiplicative Weight Algorithm}

We propose an algorithm that follows a multiplicative weight strategy (see Algorithm~\ref{alg:SMB1}). We describe here the basic approach given by the mirror descent algorithm. In Section~\ref{sec:Extension} we present an extension of the algorithm that in practice shows a better relation between the allocations and the ratios between the SLAs.
\begin{center}
\begin{algorithm}[H]
	\KwIn{Parameters $0<\varepsilon\leq \frac{1}{10}, 0  <\eta < \frac{1}{3} $.}
	Initialization: $\h_1$ any allocation over $\Delta_\varepsilon$ and $\lambda= \frac{\varepsilon^2}{8 N}$. \\
	\For{$t=1,\ldots,T$}
	{
		Set allocation $\h_t$.\\
		Read active and inactive users $A_t$ and $B_t$. $A_t^1=\{ i \in A_t : h_t(i)< \beta(i) \}$, $A_t^2= A_t\setminus A_t^1$. \\
		Set gain function $g_t(i)=\begin{cases}
		1 + \lambda & i \in A_t^1\\
		1  & i\in A_t^2\\
		0  & i \in B_t
		\end{cases}$.\\
		Update allocation:
		\begin{align*}
		\widehat{h}_{t+1}(i) & = h_t(i)e^{\eta g_t(i)}.\\
		\h_{t+1}  &= \pi_{\Delta_\varepsilon}(\widehat{h}_{t+1}).
		\end{align*}
	}
	\caption{Multiplicative Weight Update Algorithm} \label{alg:SMB1}
\end{algorithm}
\end{center}

Intuitively, the algorithm boosts active users at the expense of inactive ones, and boosts users slightly more if they are currently under their SLA. The algorithm update rule comes from a mirror descent approach applied to a Lagrangian relaxation of a work-maximization linear function at time $t$. More formally, under the assumption that active users have a huge queue, we aim to maximize the objective $\sum_{i\in A_t} w_t(i)$ subject to $w_t(i)\geq \beta(i)$ for $i\in A_t$. The update rule is obtained after applying a mirror descent with a KL-divergence distance generating function over the simplex to the Lagrangian relaxation of the previous problem (see \cite{boyd2004convex,ben2001lectures}). We restrict the projection to the truncated simplex so no user gets an allocation too close to $0$. We use  this update rule solely to guide the algorithm's decisions; however, the proofs of work maximization and SLA satisfaction do not follow from the standard analysis of mirror descent.


\section{Analysis}\label{sec:Analysis}

To give the analysis of the algorithm and prove Theorem~\ref{thm:main}, we prove the following stronger guarantees about Algorithm~\ref{alg:SMB1}. We compare its performance to the optimal offline dynamic strategy that uses at most a $1-4\varepsilon$ fraction of the resources at each time step.

\begin{theorem}\label{thm:main1}
	Given loads $L_1,\ldots,L_T$, for any $\varepsilon>0$ and $\eta >0$ such that $\varepsilon \leq  1/10$, Algorithm \ref{alg:SMB1} guarantees
	\[
	\work_{\h_1^*,\ldots,\h_{t}^*} - \work_{\text{\textrm{alg}},t} \leq 8\frac{N}{\varepsilon^2 \eta} \ln(N/\varepsilon),
	\]
for any time $1\leq t\leq T$, where $\h_1^*,\ldots,\h_T^*$ is the optimal offline sequence of $(1-4 \varepsilon )$-allocations and $\work_{\text{\textrm{alg}},t}=\sum_{i}\sum_{\tau=1}^t w_{\tau}(i)$ is the work done by Algorithm~\ref{alg:SMB1} until time $t$.
\end{theorem}

The first guarantee of Theorem~\ref{thm:main} regarding work maximization now follows simply from Theorem~\ref{thm:main1}. Given any offline dynamic policy $\h_1,\ldots, \h_T$ such that $\sum_{i}h_t(i)=1$, we define $\bar{\h}_t:= (1-4\varepsilon)\h_t$, which satisfies the assumption of Theorem~\ref{thm:main1}.  Now we have
\begin{eqnarray*}
   \work_{\text{alg}} &\geq & \work_{\bar{\h}_1,\ldots,\bar{\h}_T} - 8 \frac{N}{(\varepsilon/4)^2\eta} \ln (4N/\varepsilon) \\
   &\geq & (1-4\varepsilon)\cdot \work_{{\h}_1,\ldots,{\h}_T} - 2000 \frac{N}{\varepsilon^2 \eta} \ln(N/\varepsilon),
\end{eqnarray*}
where the first inequality follows from Theorem~\ref{thm:main1}. To argue the second, let $\w_1,\ldots,\w_T$ and $\overline{\w}_1,\ldots,\overline{\w}_T$ respectively denote the work performed by allocations $\h$ and $\bar{\h}$.
%
%
Then $(1-4\varepsilon)\w_1,\ldots,(1-4\varepsilon)\w_T$ are feasible work patterns that the allocations $\overline{\h}_1,\ldots,\overline{\h}_T$ could do, since in this setting we have $1-4\varepsilon$ capacity and the same workload. Therefore, $(1-\varepsilon)\work_{\h_1,\ldots,\h_T} \leq \work_{\overline{\h}_1,\ldots,\overline{\h}_T}$, because the users try to use their allocations at maximum.


Similarly, for SLA satisfaction, we prove a stronger bi-criteria result that implies the SLA guarantee in Theorem~\ref{thm:main}.

\begin{theorem}\label{thm:main2}
	Let $0< \varepsilon\leq 1/10$ and $0< \eta \leq  1/3$. Take any SLAs $\beta(1),\ldots,\beta(N)$ such that $\beta(i) \geq e^{\eta(1+\lambda)} \frac{\varepsilon}{N}$, where $\lambda = \frac{\varepsilon^2}{8N}$ and let $\tilde{p}=32 \frac{N^2}{\varepsilon^3 \eta} \ln (N/\varepsilon)$. Then, for any user $i$ and time $t\leq T-\widetilde{p}$, if we take $\h_1',\ldots,\h_T'$ to be allocations such that $h_t'(i)=\left(1-2\varepsilon\right)\beta(i)$, the work done by Algorithm \ref{alg:SMB1} for user $i$ satisfies
	\[
	\sum_{\tau=1}^{t+\tilde{p}} w_\tau(i) \geq \sum_{\tau = 1 }^t w_\tau'(i),
	\]
	where $\w_t'$ is the work done by the allocations $\h_1',\ldots,\h_T'$. Moreover, $\sum_{\tau=1}^t w_\tau(i) \geq \sum_{\tau=1}^t w_\tau'(i) - \beta(i)\tilde{p}$.
\end{theorem}

\subsection{The Offline Formulation}\label{subsec:LP}

\begin{figure}[t!]
	\centering
	{\small \begin{tabular}{lp{2.5in}|lp{2.5in}}
	$(P_\varepsilon)$ & $  \displaystyle \max \quad\sum_{i=1}^N \sum_{t=1}^T w_t(i) $ & $(D_\varepsilon)$ & $\displaystyle\min\quad  \sum_{i=1}^N \sum_{t=1}^TL_t(i) \gamma_t(i) + (1-\varepsilon)\sum_{t=1}^T \beta_t$ \\
	s.t. & \begin{eqnarray}
	\forall t,i & \sum_{s=1}^t w_s(i)  \leq  \sum_{s=1}^t L_s(i)   \label{eq:const_1} \\
	\forall t & \sum_{i=1}^{N} w_t(i)  \leq  1-\varepsilon   \label{eq:const_2} \\
	\forall t & \quad \w_t  \geq  0
	\end{eqnarray} & s.t. & \begin{eqnarray}
	\forall t,i &  \gamma_t(i) +\beta_t  \geq  1 \\
	\forall t &\gamma_t  \geq  \gamma_{t+1}   \\
	\forall t & \beta_t, \gamma_t  \geq  0
	\end{eqnarray}
	\end{tabular}}
	\caption{The primal and dual LP formulation for the maximum work problem.}\label{fig:LP}
\end{figure}

Before presenting the proof of Theorem~\ref{thm:main1}, we state the offline LP formulation of the maximum work problem for $(1-\varepsilon)$-allocations. We denote by $\w_t=(w_t(1),\ldots, w_t(N) )$ the work done for each user at time $t$. Given loads $L_1,\ldots,L_T$, the offline formulation and its dual LP are given in Figure \ref{fig:LP}. As written, the dual LP includes a change of variable; see Appendix \ref{sec:Omittedproofs} for details.
Constraints \eqref{eq:const_1} state that the work done for any user up to time $t$ by the allocation cannot exceed the user's loads up to that time. Constraints \eqref{eq:const_2} limit the work performed at time $t$ to at most a $1-\varepsilon$ fraction of the resource.
The LP $(D_\varepsilon)$ will be of special importance in the analysis. Using our algorithm, we will construct a dual feasible solution.

Observe that $(P_\varepsilon)$ is feasible and bounded since the feasible region is a non-empty polytope. Let $v_{P_\varepsilon}$ be the optimal value of $(P_\varepsilon)$. The following proposition gives a simple characterization of $v_{P_{\varepsilon}}$; the proof appears in Appendix~\ref{sec:Omittedproofs}.
\begin{proposition}\label{prop:optimalLP}
	$v_{P_\varepsilon} = \min_{0\leq t\leq T} \left(\sum_{s=1}^t \sum_i L_s(i) + (1-\varepsilon)(T-t)\right)$.
\end{proposition}

\subsection{Work Maximization}\label{sec:throughput}

In this section we prove Theorem~\ref{thm:main1}. 
Our first Lemma characterizes the implications of the update rule. The proof follows from a careful analysis of the dynamics using the KL-divergence and appears in Appendix~\ref{sec:Omittedproofs}.

The first result of the lemma shows the behavior of active users' allocations when the system is under-utilized ($\leq 1-\varepsilon$). In this case, all the active users receive a multiplicative boost in their allocation. The second result shows a more general behavior (see also Lemma~\ref{lem:monotonicity2}). In this case, active users with allocation below their SLA do not decrease their allocations while the other active users might decrease their allocation, but in this case, the multiplicative penalization will be less severe.

\begin{lemma}\label{lem:monotonicity}
Let $c=\frac{\varepsilon \eta}{4 N}$. Then Algorithm~\ref{alg:SMB1} satisfies the following:
\begin{enumerate}
\item 	Suppose $\sum_{i \in A_t} h_t(i) \leq  1-\varepsilon$. If $i\in A_t$, then
	$
	h_{t+1}(i) \geq h_{t}(i)(1+c).
	$

\item In general, Algorithm~\ref{alg:SMB1} satisfies $h_{t+1}(i) \geq h_t(i)$ for $i\in A_t^1$ and $h_{t+1}(i)\geq h_t(i) (1-\varepsilon c)$ for $i\in A_t^2$.
\end{enumerate}
\end{lemma}

\begin{proof}[Proof of Theorem~\ref{thm:main1}]
Given loads $L_1,\ldots,L_T\in \R_+^N$, consider the following $\{0,1\}$-matrix $M$ of dimension $N\times T$ that encodes the information about the status of queues obtained while running Algorithm~\ref{alg:SMB1}
\[
M_{i,t} = \begin{cases}
0 &i\text{'s queue is empty at }t, Q_t(i)=0,\\
1 & i\text{'s queue is not empty at }t, Q_t(i)>0.
\end{cases}
\]

Let $\widetilde{s}=\frac{\ln (N/\varepsilon)}{\varepsilon c}$, where $c$ is defined in Lemma~\ref{lem:monotonicity}. Now, pick $s^\star$ to be the maximum non-negative integer $s$ (including $0$) such that
\begin{equation}\label{eq:switch2}
\sum_{t=1}^{s} \sum_{i} L_t(i) \leq \sum_{t=1}^{s+\widetilde{s}} \sum_{i} w_t(i)
\end{equation}

\begin{claim}\label{cl:gooduser1}
	Consider any block of time $[r,r+\tilde{s}]$ where $r> s^{\star}$; then there exists a user $i$ such that $M_{i,r'}=1$ for all $r'\in [r,r+\tilde{s}]$.
\end{claim}
\begin{proof}
	Suppose not. Then we claim that $s=r$ satisfies condition \eqref{eq:switch2}. Consider any user $i$ and let $r'_i\in [r,r+\tilde{s}]$ be such that $M_{i,r'_i}=0$. Then work done by the user $i$ up to time $r+\tilde{s}$ is at least
	$$\sum_{t=1}^{r+\tilde{s}}  w_t(i)\geq \sum_{t=1}^{r'_i}  w_t(i)= \sum_{t=1}^{r_i'}  L_t(i)\geq \sum_{t=1}^{r}  L_t(i).$$
	Now summing over all $i$, we get the desired contradiction.
\end{proof}

We now prove the following claim that shows that the algorithm ensures that, on average, the total resource utilization after $s^\star$ is close to $1-4\varepsilon$. The proof of the Claim relies on Lemma~\ref{lem:monotonicity} and appears in Appendix~\ref{sec:Omittedproofs}.
\begin{claim}\label{cl:gooduser2}
	Let $B=[r,r+\tilde{s})$ with $r>s^\star$ be a consecutive block of $\tilde{s}$ timesteps and let $B'=\{  t \in B :  \sum_{j\in A_t} h_t(j) \leq 1- \varepsilon  \}$ be the time steps in $B$ with low utilization. Then $|B'|\leq 4\varepsilon \tilde{s}$ and therefore,  $\sum_{t=s^\star +1}^T \sum_i w_t(i)  \geq ( 1 - 4 \varepsilon )(T-s^\star) -\widetilde{s}.$
\end{claim}

Now, consider the following feasible dual solution of $(D_{4\varepsilon})$: $\gamma_t(i)=1,\beta_t=0$ for all users $i$ and $t=1,\ldots,s^\star$, and $\gamma_t(i)=0, \beta_t=1$ for all users $i$ and $t=s^\star+1,\ldots,T$. Observe that $\sum_{t=1}^T \beta_t = T-s^\star$. For optimal $(1-4 \varepsilon)$-allocations $\h_1^*\ldots,\h_T^*$ we obtain
	\begin{align*}
	\work_{\h_1^*,\ldots,\h_T^*} & \leq v_\text{dual}(\gamma_1,\ldots,\gamma_T,\beta_1,\ldots\beta_T) \tag{weak duality} \\
	&= \sum_{t=1}^{s^\star}\sum_i L_t(i) + (1-4 \varepsilon)(T-s^\star)\\
	& \leq \sum_{t=1}^{s^\star + \widetilde{s}} \sum_i w_t(i) + (1-4 \varepsilon)(T-s^\star) \tag{choice of $s^\star$}\\
	& \leq \sum_{t=1}^{s^{\star}}  \sum_i w_t(i) + \widetilde{s} + \sum_{t=s^{\star}+1}^{T}\sum_i w_t(i) + \widetilde{s} \tag{Claim \ref{cl:gooduser2}}\\
	& = \work_{\text{alg}} + 8\frac{N}{\varepsilon^2 \eta} \ln(N/\varepsilon).
	\end{align*}
where we have used $\sum_{t=s^\star+1}^{s^\star+\widetilde{s}} w_t(i) \leq \widetilde{s}$ and the definition of $\widetilde{s}$. 
\end{proof}


\subsection{SLA Satisfaction}\label{sec:SLAsat}

In this section, we prove Theorem~\ref{thm:main2}. Recall that $\lambda = \frac{\varepsilon^2}{8N}$ and $A_t^1=\{ i\in A_t: h_t(i)<\beta(i)  \}$ is the set of active users receiving less than their SLAs and $A_t^2=A_t\setminus A_t^1$ is the set of active user receiving at least their SLA. Analogous to Lemma~\ref{lem:monotonicity}, we have the following lemma, whose proof appears in Appendix~\ref{sec:Omittedproofs}.

\begin{lemma}\label{lem:monotonicity2}
	Assume $\varepsilon \leq \frac{1}{10}$, $\eta\leq \frac{1}{3}$ and $\beta(i)\geq 2\frac{\varepsilon}{N}$ for all users. Then for any $i\in A_t^1$, Algorithm~\ref{alg:SMB1} guarantees $h_{t+1}(i)\geq h_t(i)(1+c')$, where $c'= \frac{\varepsilon\eta \lambda}{2N}$.
\end{lemma}

\begin{proof}[Proof of Theorem \ref{thm:main2}]
	Let $\widetilde{p}= \left\lceil\frac{\ln(N/\varepsilon)}{\ln(1+c')} \right\rceil$, where $c'$ is defined in Lemma~\ref{lem:monotonicity2}. Now, we proceed by induction on $t$ to prove that $\sum_{\tau=1}^{t+\widetilde{p}} w_\tau(i) \geq \sum_{\tau=1}^t w_\tau'(i)$, where $\w_t'$ is the work done by the allocations $\h_1',\ldots,\h_T'$. Clearly, the case $t=0$ is direct.
	
	Take $t\geq 1$ and suppose the result is true for $t-1$. If there exists $r\in[t, t+\widetilde{p}]$ such that user $i$'s queue is empty, then
	\[
	\sum_{\tau=1}^{t+\widetilde{p}} w_\tau(i) \geq \sum_{\tau=1}^{r} w_{\tau}(i) = \sum_{\tau=1}^{r}  L_\tau(i) \geq \sum_{\tau=1}^t w_\tau'(i).
	\]
	Therefore, assume that for all $\tau \in [t,t+\widetilde{p}]$ we have that user $i$'s queue is non-empty. By the induction hypothesis
	\[
	\sum_{\tau=1}^{t-1+\widetilde{p}} w_t(i) \geq  \sum_{\tau=1}^{t-1} w_\tau'(i).
	\]
	In order to complete the proof, we need to prove that  $w_{t+\widetilde{p}}(i) \geq w_t'(i)$. We proceed as follows. Suppose that for all $\tau\geq t$ we have $ w_\tau(i)< (1-\varepsilon)\beta(i)$. By Lemma~\ref{lem:monotonicity2}, at each time $\tau \in [t,t+\widetilde{p}]$ the allocation of user $i$ increases multiplicatively by a rate $(1+c')$. Therefore,
	\[
	w_{t+\widetilde{p}}(i) \geq \frac{\varepsilon}{N}(1+c')^{\widetilde{p}} \geq 1 \geq \beta(i),
	\]
	a contradiction. From the previous analysis we obtain the existence of $\tau^\star\in [t,t+\widetilde{p}]$ such that $w_{\tau^\star}(i)\geq (1-\varepsilon) \beta(i)$. By using Lemmas \ref{lem:monotonicity} and \ref{lem:monotonicity2}, we can show that the allocation $h_\tau(i)$ will never go below $(1-\varepsilon c)(1-\varepsilon)\beta(i)$ for all $\tau\geq \tau^\star$. In particular $w_{t+\widetilde{p}}(i)\geq (1-\varepsilon c)(1-\varepsilon)\beta(i) \geq (1-2\varepsilon)\beta(i) \geq w_t'(i)$.
\end{proof}

%
%

\subsection{Extension to Proportionality and Over-commitment}\label{sec:Extension}

In the previous subsections, we have introduced the first version of the multiplicative weights algorithm. We explained how we deduced our algorithm using mirror descent and proved its theoretical guarantees. Even though Algorithm~\ref{alg:SMB1} guarantees individual SLA satisfaction, this simple policy can lead to undesirable results that do not respect the ratio between allocations. If one user has an SLA twice the size of another, it would be reasonable for the former to expect allocations at least twice as big as the latter's if both are consistently busy. Likewise, the second user would expect allocations no less than half of the first user's. In other words, both users should expect shares that respect the ratio between their SLAs.

To illustrate this unsatisfactory behavior in Algorithm~\ref{alg:SMB1}, we run it with three users having SLAs $\beta(1)=0.5$, $\beta(2)=0.3$ and $\beta(3)=0.2$. We set $\eta\leq 1/3$ and $\varepsilon\leq 1/10$. For simplicity, the initial allocation will be uniform. In our example, user $1$ is always idle. User $2$ consistently demands $1$ unit of resource. User $3$ begins idle and remains so until user $2$'s allocation reaches $1-\varepsilon$. This takes roughly $\frac{1}{\eta \varepsilon}$ time steps; call this time $t_0$. Starting at time $t_0$, user $3$ demands unit loads every time step for the rest of the horizon. Initially, the allocations are uniformly $1/3$ for everyone. Between time $1$ and $t_0$, user $2$'s allocation increases until it hits $1-\varepsilon$, since they are the only active user. After $t_0$, user $3$ becomes active, and has an allocation below their SLA. Therefore, the algorithm redistributes allocation from user $2$ to $3$ until user $3$'s allocation hits $0.2$. After this, allocations remain stable at approximately $h_t(1)= \frac{\varepsilon}{3}$, $h_t(2)= 0.8- \frac{\varepsilon}{3}$ and $h_t(3)=0.2$. User $2$ receives about $4$ times the allocation of user $3$ if $\varepsilon$ is small enough. However, a better allocation for user $2$ and $3$ is $\frac{\beta(2)}{\beta(2)+\beta(3)} = \frac{3}{5}$ and $\frac{\beta(3)}{\beta(2) + \beta(3)} = \frac{2}{5}$ respectively. These allocations reflect the ratio $\frac{\beta(2)}{\beta(3)}$ between active users.

Given this, we propose a slight modification of Algorithm~\ref{alg:SMB1}, shown in Algorithm~\ref{alg:SMB2}. As before, the plan is always to benefit active users. However, this time, we boost active users slightly more if they fall behind their ``proportional SLA'' among active users. Intuitively, if there are $n<N$ active users for a long period of time, the allocation of these active users should converge to their proportional share.
\begin{center}
	\begin{algorithm}[H]
		\KwIn{Parameters $0<\varepsilon\leq \frac{1}{10}, 0  <\eta < \frac{1}{3} $.}
		Initialization: $\h_1$ any allocation over $\Delta_\varepsilon$ and $\lambda= \frac{\varepsilon^2}{8 N}$. \\
		\For{$t=1,\ldots,T$}
		{
			Set allocation $\h_t$.\\
			Read active and inactive users $A_t$ and $B_t$. $A_t^1=\left\{ i \in A_t : h_t(i)< (1-\varepsilon)\frac{\beta(i)}{\sum_{j\in A_t} \beta(j)} \right\}$, $A_t^2= A_t\setminus A_t^1$. \\
			Set gain function $g_t(i)=\begin{cases}
			1 + \lambda & i \in A_t^1\\
			1  & i\in A_t^2\\
			0  & i \in B_t
			\end{cases}$.\\
			Update allocation:
			$\widehat{h}_{t+1}(i)  = h_t(i)e^{\eta g_t(i)}, \forall i $ and $\h_{t+1} = \pi_{\Delta_\varepsilon}(\widehat{h}_{t+1})$
		}
		\caption{Extended Multiplicative Weight Update Algorithm} \label{alg:SMB2}
	\end{algorithm}
\end{center}

For technical reasons, the set $A_t^1$, the active users with allocation below their proportional share among active users at time $t$, has to be defined as $\left\{ i \in A_t : h_t(i)< (1-\varepsilon)\frac{\beta(i)}{\sum_{j\in A_t} \beta(j)} \right\}$. The reason behind this choice is to ensure that if $A_t^1\neq \emptyset$ and the resource is nearly fully utilized, i.e., $\sum_{i\in A_t} h_t(i)\geq 1-\varepsilon$, then there is a different active user $j\neq i$ from which we can move allocation to $i$. This is fundamental in the proof of Theorem~\ref{thm:SLAextension} below.

In terms of work maximization and SLA satisfaction, Algorithm~\ref{alg:SMB2} provides exactly the same guarantees as Algorithm~\ref{alg:SMB1}.

\begin{theorem}\label{thm:mainextension}
	Given loads $L_1,\ldots,L_T$, for any $\varepsilon>0$ and $\eta >0$ such that $\varepsilon \leq 1/10$, Algorithm \ref{alg:SMB2} guarantees
	\[
	\work_{\h_1^*,\ldots,\h_t^*} - \work_{\text{alg},t} \leq 8\frac{N}{\varepsilon^2 \eta} \ln(N/\varepsilon),
	\]
for any time $1\leq t\leq T$, where $\h_1^*,\ldots,\h_T^*$ is an optimal offline sequence of $(1-4 \varepsilon )$-allocations and $\work_{\text{alg},t}=\sum_{i}\sum_{\tau=1}^t w_{\tau}(i)$ is the overall work done by Algorithm~\ref{alg:SMB2} until time $t$.
\end{theorem}

The proof of Theorem~\ref{thm:mainextension} is exactly the same as the proof of Theorem~\ref{thm:main1}. To see this, observe that the proof of Theorem~\ref{thm:main1} only uses the fact that the allocations of \emph{every} active user get a multiplicative boost whenever the usage is below $1-\varepsilon$. The last statement is true since Lemma~\ref{lem:monotonicity} also holds in this case.

For SLA satisfaction, we have the following stronger statement.
\begin{theorem}\label{thm:SLAextension}
	Let $0< \varepsilon\leq 1/10$, $0< \eta \leq 1/3$, $\lambda = \frac{\varepsilon^2}{8N}$ and $\tilde{p}=32 \frac{N^2}{\varepsilon^3 \eta} \ln (N/\varepsilon)$. Take any SLAs $\beta(1),\ldots,\beta(N)$ such that $\frac{\beta(i)}{\sum_k \beta(k)} \geq \frac{e^{\eta(1+\lambda)}}{1-\varepsilon} \frac{\varepsilon}{N}$. Then, for any user $i$ and time $t$, if we take $\h_1',\ldots,\h_T'$ to be the allocations such that $ h_t'(i)=\left(1-2\varepsilon\right) \frac{\beta(i)}{\sum_k \beta(k)}$, the work done by Algorithm \ref{alg:SMB1} for user $i$ satisfies
	\[
	\sum_{\tau=1}^{t+\tilde{p}} w_\tau(i) \geq \sum_{\tau = 1 }^t w_\tau'(i),
	\]
where $\w_t'$ is the work done by the allocations $\h_1',\ldots,\h_T'$. Moreover $\sum_{\tau=1}^t w_\tau(i) \geq \sum_{\tau=1}^t w_\tau'(i) - \frac{\beta(i)}{\sum_k \beta(k)}\tilde{p}$.
\end{theorem}
The proof of this result is similar to the proof of Theorem~\ref{thm:main2}. A subtle difference is that the analogue of Lemma~\ref{lem:monotonicity2} holds if we add the hypothesis $\sum_{i\in A_t} h_t(i) > 1-\varepsilon$. We skip the proof for brevity.
\begin{lemma}\label{lem:monotonicity3}
	Assume $\varepsilon\leq \frac{1}{10}$, $\eta \leq \frac{1}{3}$ and $\frac{\beta(i)}{\sum_k \beta(k)} \geq \frac{e^{\eta(1+\lambda)}}{1-\varepsilon} \frac{\varepsilon}{N}$ for all users. In Algorithm~\ref{alg:SMB2}, if $\sum_{k\in A_t} h_t(k)> 1- \varepsilon$ then for any $i\in A_t^1$ we have $h_{t+1}(i) \geq (1+c') h_t(i)$, where $c'= \frac{\varepsilon \eta \lambda}{2N}$.
\end{lemma}

Lemma~\ref{lem:monotonicity} and \ref{lem:monotonicity3} ensure that any active user gets a multiplicative boost of at least $(1+c')$. Therefore, any user that is active $\widetilde{p}$ consecutive times will have an allocation of at least $(1-2\varepsilon) \frac{\beta(i)}{\sum_k \beta(k)}$. Then, by following the same inductive proof of Theorem~\ref{thm:main2}, we obtain Theorem~\ref{thm:SLAextension}.

If the resource is not over-committed, $\sum_k\beta (k)\leq 1$, this result implies that Algorithm~\ref{alg:SMB2} ensures for each user $i$ an amount of work comparable with $\frac{\beta(i)}{\sum_k \beta(k)} \geq \beta(i)$; that is, we obtain the standard SLA satisfaction guarantee.
In the over-commitment regime, $\sum_{i=1}^N \beta(i)> 1$, we do retain some performance guarantees.  The update according to almost-normalized SLAs in Algorithm~\ref{alg:SMB2} still works and Theorem~\ref{thm:mainextension}'s work maximization guarantee still applies, as its proof does not rely on the SLAs. On the other hand, Theorem~\ref{thm:SLAextension} states that individually, each user does work comparable to their normalized SLA. If the level of over-commitment is not large, each user is still guaranteed service ``almost'' at their SLA; for example, if the resource is over-committed by $10\%$, each user receives service comparable to $1.1^{-1} \approx 91\%$ of their SLA.

Another interesting byproduct of the work maximization guarantee is the following result.
\begin{corollary}
Under the assumptions of Theorem~\ref{thm:main2}, 
suppose there is a time $1\leq \tau\leq T$ with $\sum_{t=1}^\tau \sum_{i} w_\tau'(i)= \sum_{t=1}^\tau \sum_{i} L_t(i)$; i.e.\ the optimal offline $(1-4 \varepsilon)$-allocation is able to finish all work up until $\tau$. Then, the sum of queue lengths at time $\tau$ induced by Algorithm \ref{alg:SMB2} is at most $8\frac{N}{\varepsilon^2 \eta} \ln(N/\varepsilon)$. 
In particular, at time $\tau$ each user's queue length is at most this value. 
\end{corollary}

In practical settings, it is commonplace to assume an arrival rate is lower than the work processing rate. In stochastic settings, stationary states cannot be achieved without this assumption; see, e.g.,\ \cite{tassiulas1990stability}. In our context, this can be reinterpreted as having times within the operating horizon where the optimal offline solution is able to finish all work arriving up until that time. At these particular times, the corollary guarantees that Algorithm \ref{alg:SMB2}'s queue lengths are constant. In other words, the algorithm does not starve individual users to achieve work maximization and keeps their queues short, an appealing property in cloud systems. 

\section{Experiments}\label{sec:Numerical}

In this section, we empirically test Algorithm~\ref{alg:SMB2} against a family of offline and online algorithms. We aim to measure the performance on both synthetic data as well as real CPU traces from a production service in Microsoft's cloud. We quantify the performance on the following three criteria.


\begin{itemize}
	\item \textbf{Work maximization.} We compare the overall work done by Algorithm~\ref{alg:SMB2} against various benchmark algorithms.
	
	\item \textbf{SLA guarantee.} We examine the extent to which our algorithm achieves the cumulative work of the static SLA policy for each user. We do so by measuring the cumulative work over plausible time windows.

	
	\item \textbf{Queue length.} We compare the $2$-norm of the individual queues over time. 
 We use this metric as a proxy for the system latency, which is not captured by our theoretical results. 
\end{itemize}


We consider the following \emph{online algorithms}, against which we benchmark our algorithm.
\begin{itemize}
\item \textbf{Static SLA Policy (Static)}. Each user gets their SLA as a constant, static allocation. We call this algorithm Static.

\item \textbf{Proportional Online (PO)}. In each iteration, every active user will get their SLA normalized by the sum of SLAs of active users (just their SLA if there are no active users). This simple algorithm seems suitable for a practical implementation; however, its performance can be arbitrarily bad in terms of work maximization. The formal description appears in Algorithm~\ref{alg:OnlProp} in Appendix~\ref{sec:Add_Alg}. We call this Algorithm PO.

\item \textbf{Online Work Maximizing (OWM)}. This algorithm divides users into three categories: $A,B$ and $I$ (active users with allocation, actives users without allocation and inactive users). At each iteration, the resource is divided uniformly among users in $A$. If a user in $A$ becomes inactive, they are moved to $I$. If a user in $I$ becomes active, they are moved to $B$. When $A$ becomes empty, we move all users from $B$ to $A$. In Appendix~\ref{sec:localgreedy} we prove that this method is work maximizing. However, this greedy strategy is not guaranteed to satisfy SLA constraints for general input loads. We call this algorithm OWM.
\end{itemize}
We also consider the following \emph{offline algorithms}, against which we benchmark our algorithm. 
\begin{itemize}
\item[*] \textbf{Optimal $1$-allocations (PG)}. The optimal offline solution to the work maximization problem. This solution is computed using Algorithm~\ref{alg:PropGreedy}, which we call Proportional Greedy (PG). This algorithm can be considered as the offline counterpart of Proportional Online.

\item[*] \textbf{Optimal $(1-\varepsilon)$-allocation (restPG)}. Offline solution to the work maximization problem with resource restricted to $1-\varepsilon$, where $\varepsilon$ is the parameter of Algorithm~\ref{alg:SMB2}. 
\end{itemize}

\subsection{Synthetic Experiment}

\subsubsection{Description of the Experiment}~
In this experiment, we consider a synthetically generated input sequence, which we use to examine how online algorithms adapt to different conditions. 
Specifically, our system consists of three users with SLAs of $\beta(1) = 0.2$, $\beta(2)=0.3$ and $\beta(3) = 0.5$. We consider a time horizon of $T=3,\!000,\!000$. The load input sequence is divided into six periods: $P_i=\left[ \frac{(i-1)T}{6},  \frac{iT}{6} \right)$ for $i=1,\ldots,6$. In each period, only two users demand new resources. During the first $3$ periods, the random demand has a mean proportional to the users' SLA. In the following $3$ periods, the random demand changes to a distribution with uniform mean among users demanding resources. Specifically:
\begin{itemize}
	
\item During $P_1$, only users $2$ and $3$ demand the following loads. At the beginning of $P_1$, 
i.e., $t=1$, user $2$ demands a large load of $L_1(2) \sim \frac{T}{6}\cdot \mathrm{Gamma}\left(2000, \frac{1}{2000}\cdot\frac{\beta(2)}{\beta(2)+ \beta(3)} \right)$ and $L_1(3)=0$. During the rest of period $P_1$, $L_t(2)=0$ and $L_t(3) 
\sim \mathrm{Gamma}\left(2000, \frac{1}{2000}\cdot\frac{\beta(3)}{\beta(2)+ \beta(3)} \right)$. User $1$ demands nothing during this entire period. Similar loads are set for period $P_2$ and $P_3$. 
	
\item Similarly, during $P_4$ users $2$ and $3$ demand $L_t(i) \sim \mathrm{Gamma}\left(2000, \frac{1}{2000}\cdot \frac{1}{2}\right)$ for $i=2,3$. During $P_5$ users $1$ and $2$ demand $L_t(i) \sim \mathrm{Gamma}\left(2000, \frac{1}{2000}\cdot \frac{1}{2}\right)$ for $i=1,2$. During $P_6$ users $1$ and $3$ demand $L_t(i) \sim \mathrm{Gamma}\left(2000, \frac{1}{2000}\cdot \frac{1}{2}\right)$ for $i=1,3$.
	
\end{itemize}

%

The expectation of a Gamma$(k,\theta)$ random variable is given by $k\theta$ and the variance is given by $k\theta^2$ (see e.g.\ \cite{feller1957introduction}). For example, in period $P_1$ user $2$'s expected load is $ \frac{T}{6} \frac{\beta(2)}{\beta(2)+\beta(3)}$, with variance $\frac{1}{2000} \left( \frac{\beta(2)}{\beta(2)+\beta(3)} \right)^2$. Similarly, user $3$'s expected total load is $ \frac{T}{6} \frac{\beta(3)}{\beta(2)+\beta(3)}$. Thus, the expected overall load is $T/6$, exactly the length of the period. 
 A brief summary of Gamma distribution's properties is given in Appendix \ref{sec:Gamma}.





We instantiate Algorithm~\ref{alg:SMB2} with $\eta=\frac{1}{3}$, $\varepsilon=0.02$ and $T=3,\!000,\!000$.

\subsubsection{Results}

\textbf{Work maximization.} In Figure~\ref{fig:char3users} we present the cumulative work difference between PG and Algorithm~\ref{alg:SMB2} (solid blue line with star), restPG and Algorithm~\ref{alg:SMB2} (red dashed line with triangle), PO and Algorithm~\ref{alg:SMB2} (solid magenta line), Static and Algorithm~\ref{alg:SMB2} (solid green line with small circle) and OWM and Algorithm~\ref{alg:SMB2} (solid cyan line with large circle). Intuitively, one positive unit of difference implies the corresponding algorithm is ahead of Algorithm~\ref{alg:SMB2} by one unit of time.

First, we consider the comparison to online algorithms  Static, PO and OWM.  Algorithm~\ref{alg:SMB2} outperforms Static significantly, by roughly $700,\!000$ units of time. During the first half of the experiment PO shows good performance, but in the second half of the experiment (when the load distribution changes), Algorithm~\ref{alg:SMB2} outperforms PO. 
This shows that Algorithm~\ref{alg:SMB2} can adapt to changing input sequences that PO cannot adapt to. Finally, OWM surpasses Algorithm~\ref{alg:SMB2} during the whole experiment, with a performance similar to PG; 
this is an expected result since OWM is work maximizing.

With respect to offline algorithms, PG outperforms Algorithm~\ref{alg:SMB2} by roughly $10,\!000$ time units. On the other hand, Algorithm~\ref{alg:SMB2} surpasses restPG by approximately $20,\!000$ units. This shows that Algorithm~\ref{alg:SMB2} indeed performs better than the offline optimum with a slightly reduced amount of resources, as guaranteed by the theoretical results.
\begin{figure}[h!]
	\centering
		\centering
	\includegraphics[width=0.75\linewidth]{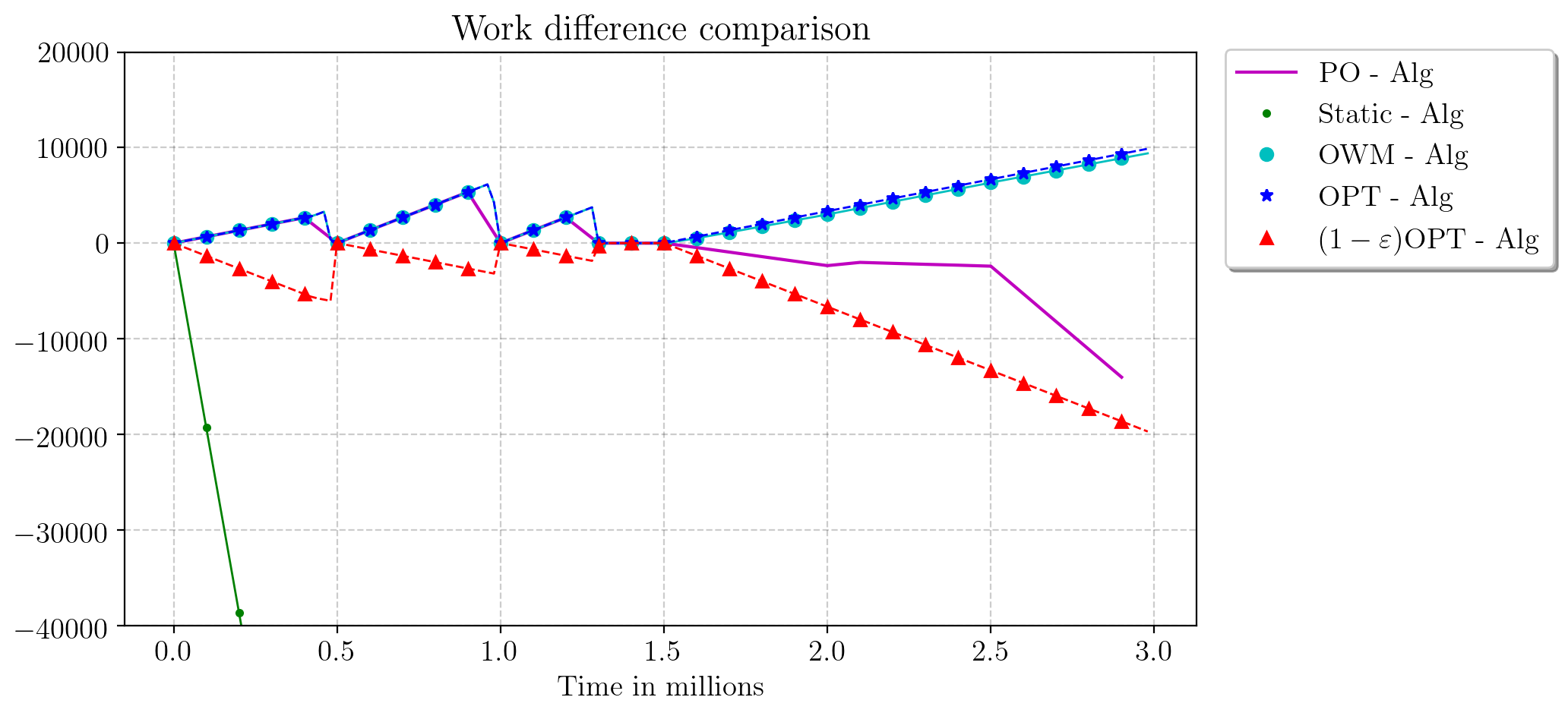}
	\caption{Difference between algorithms' work. Alg is short for Algorithm~\ref{alg:SMB2}. Observe that the differences ``OPT - Alg'' ``and ``OWM - Alg'' slightly overlap.}
	\label{fig:char3users}
\end{figure}

\textbf{SLA satisfaction.} Among online algorithms, Static and PO satisfy SLA restrictions by design, but as seen earlier they are not competitive in terms of work maximization. We focus on the comparison between OWM and Algorithm~\ref{alg:SMB2} as far as SLA satisfaction is concerned. Although OWM performs extremely well in work maximization, this comes at a significant price in SLA satisfaction. In Figure~\ref{fig:work} we depict empirically this behavior by plotting the instantaneous work done by users $2$ and $3$ by OWM and Algorithm~\ref{alg:SMB2} during period $P_1$. We empirically observe that Algorithm~\ref{alg:SMB2} approximately satisfies user $3$'s SLA,  but OWM does not allocate the user any resources. Such an extreme behavior arises because OWM is geared towards work maximization rather than SLA satisfaction.
%
%
\begin{figure}[h!]
	\centering
	\begin{minipage}[l]{0.7\linewidth}
		\includegraphics[width=\textwidth]{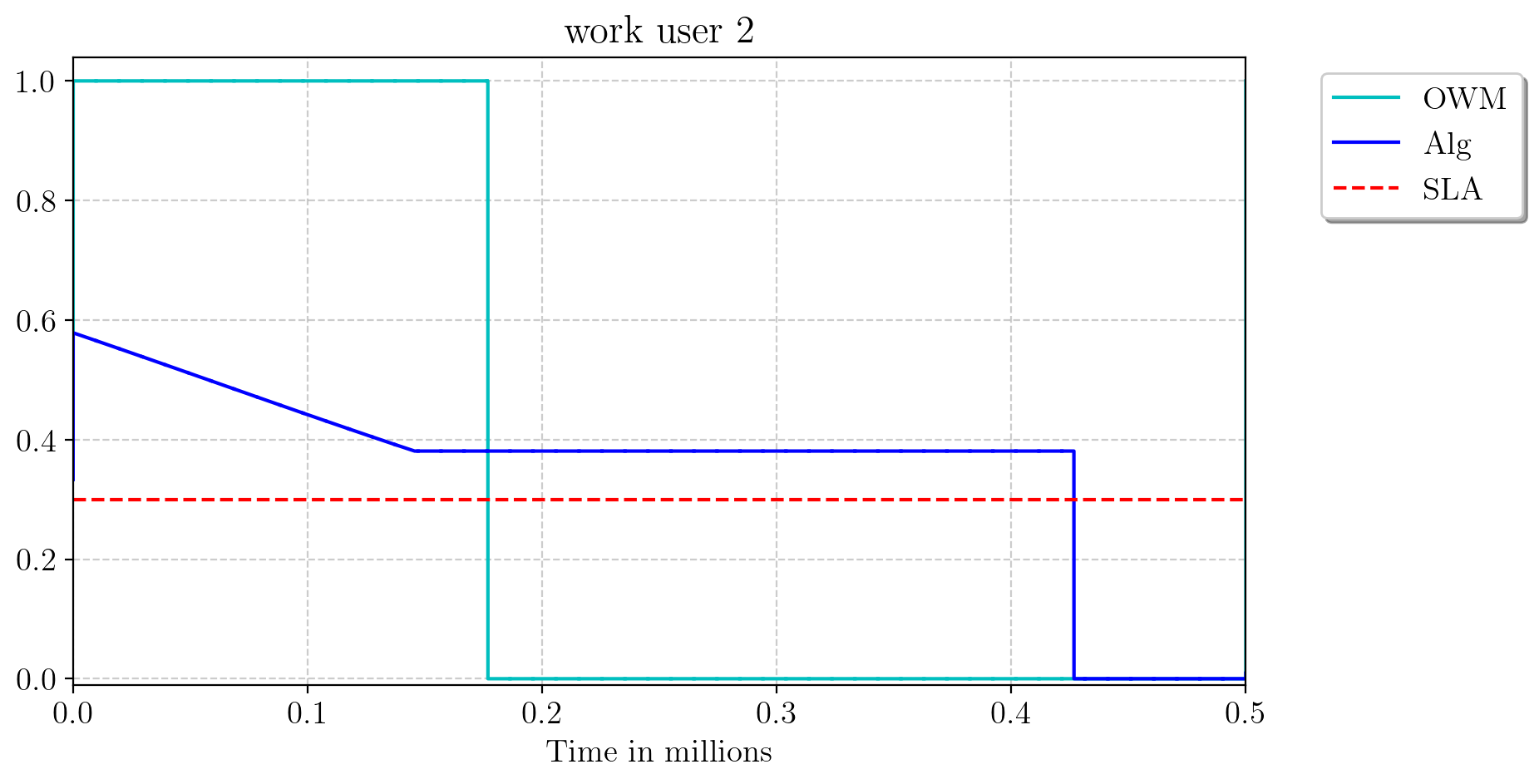}
	\end{minipage}\\
	\begin{minipage}[l]{0.7\linewidth}
	\includegraphics[width=\textwidth]{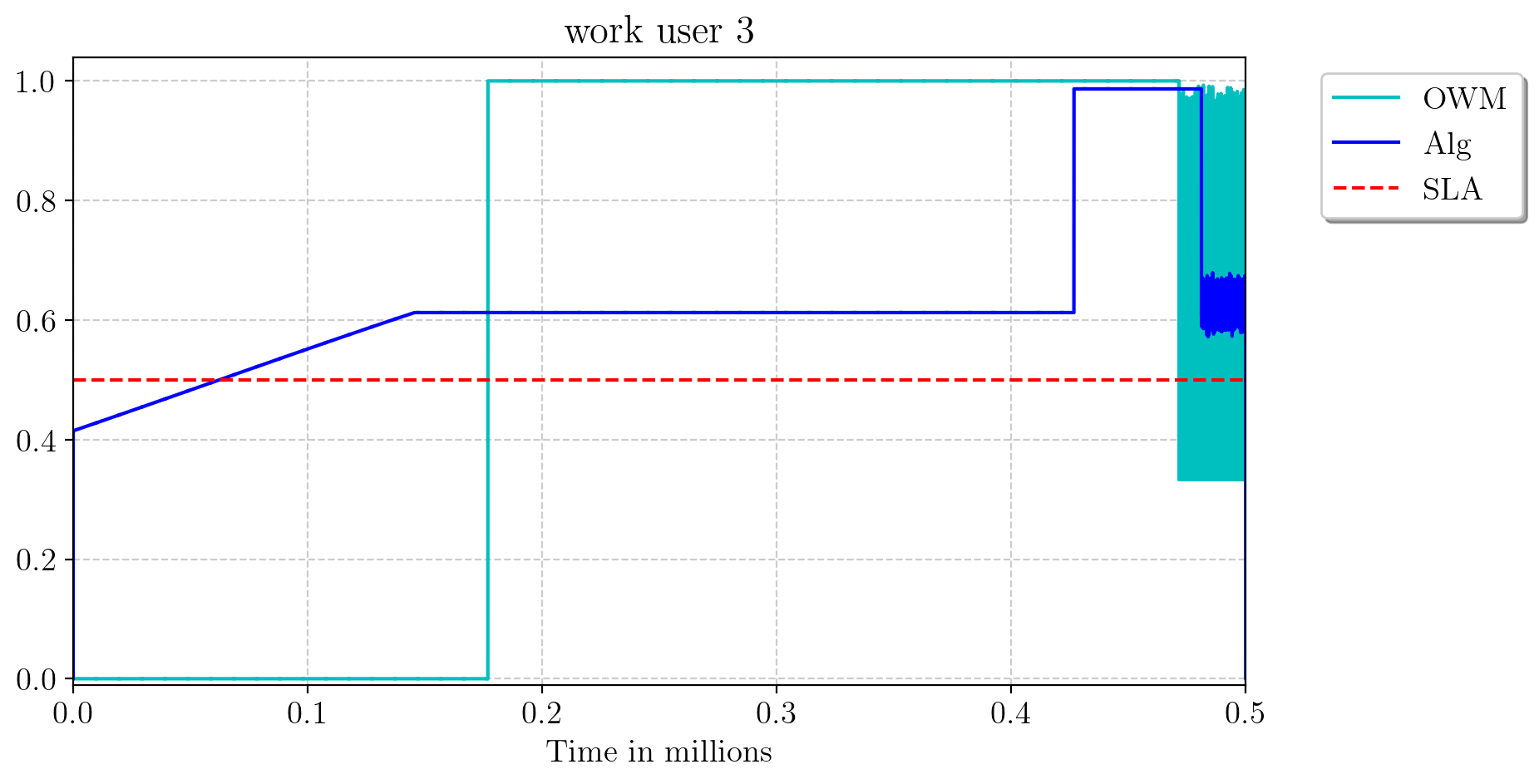}
\end{minipage}\\
	\caption{Instantaneous work for user $2$ and $3$ during period $P_1$. User $3$ does not receive any allocation in OWM until user $2$ finishes all of their work.}
	\label{fig:work}
\end{figure}

\textbf{Queue lengths.} In Figure~\ref{fig:queuessquared} we present the $2$-norm of queues induced by Algorithm~\ref{alg:SMB2} (solid blue), Static (solid magenta with star), PO (solid cyan with large circle) and OWM (solid black with triangle). Once again, we can interpret one unit of norm as one unit of latency. Experimentally, we observe that Static shows the worst performance with a final $2$-norm of $524,\!000$ units. Algorithm~\ref{alg:SMB2} ends with a $2$-norm of $10,\!000$ units, PO with $26,\!970$ units, and OWM with $381$ units. As remarked above, even though OWM induces very small queues, this comes at the cost of not satisfying SLA requirements. 
\begin{figure}[h!]
		\centering
	\includegraphics[width=0.75\linewidth]{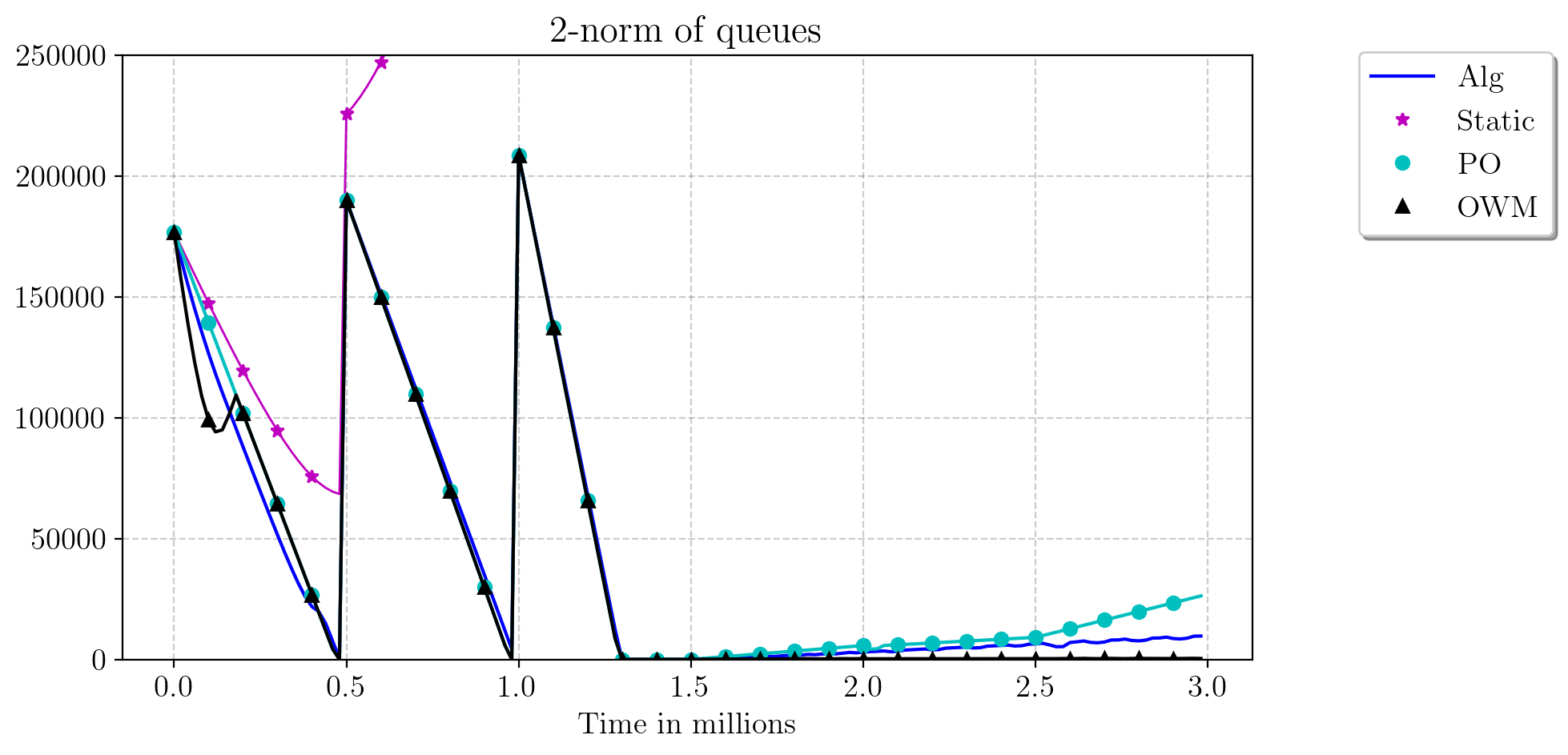}
	\caption{$2$-norm of queues.}
	\label{fig:queuessquared}
\end{figure}

{\textbf{Summary.} Among all online algorithms, Algorithm~\ref{alg:SMB2} is able to best balance work maximization and SLA satisfaction. In particular, Algorithm~\ref{alg:SMB2} is only slightly worse in terms of work maximization compared to OWM; this is expected since Algorithm~\ref{alg:SMB2} always reserves a small fraction of the resource for each user, regardless of activity. Furthermore, our results show that the actual total work done by Algorithm~\ref{alg:SMB2} is much better than the theoretical guarantee of Theorem~\ref{thm:mainextension}, as it substantially outperforms $(1-\varepsilon)$OPT (restPG). Finally, we observe small queues throughout the entire horizon, which is key for maintaining reasonable latencies.}

\subsection{Experiment with Real Data}


We next describe the results of our computational study using real world data. For these experiments, we obtained CPU traces of a production service on Azure, Microsoft's public cloud. The data consists of demand traces of six different users over a time window of approximately ten days.

To show Algorithm~\ref{alg:SMB2}'s robustness with respect to (short-term) real data, we also consider the following measurement:
\begin{itemize}
%
\item \textbf{Instantaneous SLA.} We focus on a modified SLA satisfaction criterion because of the relatively short horizon (about $14,000$ minutes); 
we assess Algorithm~\ref{alg:SMB2}'s performance in the following way. For a user $i$ and any time $t$, we compare the cumulative work done by user $i$ in Algorithm~\ref{alg:SMB2} during a time window $[t,t+\tau)$ versus the cumulative work done by the same user $i$ under a Static SLA Policy during the time window $[t,t+\tau)$. In order to have a meaningful comparison, at time $t$, we run the Static SLA Policy with the queues of Algorithm~\ref{alg:SMB2} at time $t$. The motivating question is, \emph{what happens if at time $t$ and the next $\tau$ time steps we run the static policy instead of Algorithm~\ref{alg:SMB2}?} For this experiment, we used $\tau = 500$ minutes.
	
	%
%

\end{itemize}

The data set consists of demand traces of six users of exactly $14,\!628$ minutes (approximately ten days). Each user is assigned their normalized average workload as SLA. (Since the data is proprietary, we cannot disclose actual SLAs.) For the purpose of the experiment, we run Algorithm~\ref{alg:SMB2} with parameters $\varepsilon=0.01$ and $\eta = \frac{1}{3}$.

\subsubsection{Results}

\textbf{Work maximization.} We depict in Figure~\ref{fig:MS_work_users_diff} the following differences: cumulative work done until time $t$ by optimal $1$-allocations (PG) and Algorithm~\ref{alg:SMB2}, $(1-\varepsilon)$-allocations (restPG) and Algorithm~\ref{alg:SMB2}, PO and Algorithm~\ref{alg:SMB2}, Static and Algorithm~\ref{alg:SMB2}, and OWM and Algorithm~\ref{alg:SMB2}. In a similar fashion to the previous experiment, one positive unit can be interpreted as Algorithm~\ref{alg:SMB2} being one unit (minute) of work behind, and one negative unit means Algorithm~\ref{alg:SMB2} is ahead by a minute.

We observe that Algorithm~\ref{alg:SMB2} outperforms all online benchmarks. Against Static, the final difference is $105$ units, with a maximum difference of $167$ units. Against PO the final difference is $25$ units with a maximum difference of $37$. Finally, against OWM, the final difference is $20$ with a maximum difference of $21$. Surprisingly, for this data set Algorithm~\ref{alg:SMB2} is able to surpass even OWM.

Regarding the offline algorithms, PG surpasses Algorithm~\ref{alg:SMB2} during the whole experiment as expected, with a final difference of $14$ units. On the other hand, Algorithm~\ref{alg:SMB2} outperforms restPG by $26$ units by the end of the experiment.
\begin{figure}[h!]
	\centering
	\includegraphics[width=0.75\linewidth]{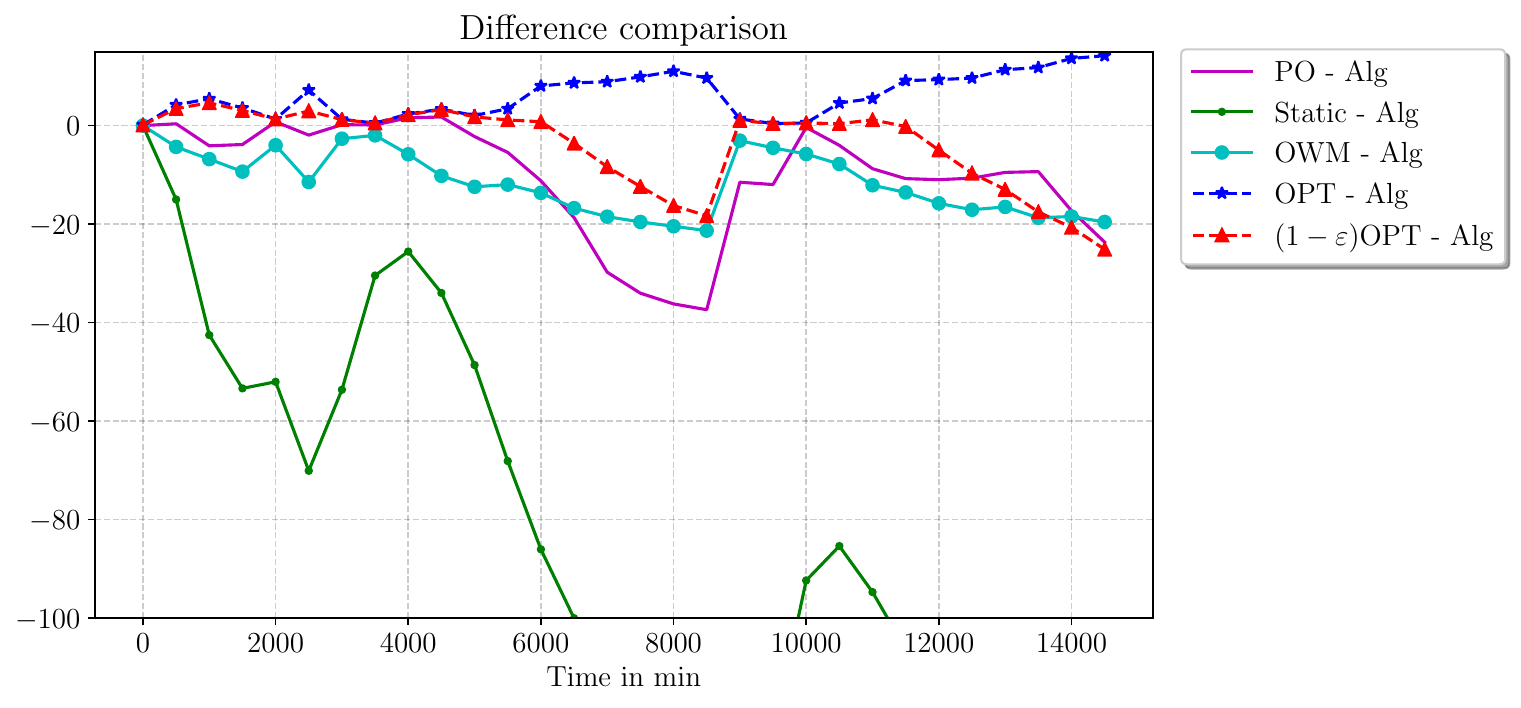}
	\caption{Difference of cumulative works.}
	\label{fig:MS_work_users_diff}
\end{figure}

\textbf{Instantaneous SLA.} In Table~\ref{tab:MS_diff_SLA} we report statistics on the differences between the cumulative work done in time windows $[t,t+\tau)$  by Static and Algorithm~\ref{alg:SMB2} for each user. For each $t$ and each user $i$, the exact formula is $r_i(t)={\sum_{r=t}^{t+\tau -1} w_r'(i)}-{\sum_{r=t}^{t+\tau -1} w_r(i)}$, where $w_r(i)$ is the work done by user $i$ under Algorithm~\ref{alg:SMB2} and $w_r'(i)$ is the work done by user $i$ with Static (with queues at $t$ given by Algorithm~\ref{alg:SMB2}). The table shows the min, max, average and standard deviation of $\{  r_i(t)\}_t$ when $\tau=500$ minutes. A positive unit means Algorithm~\ref{alg:SMB2} is outperformed by Static during $[t,t+\tau)$ under the same initial conditions by one unit of time. In general, we observe that all users show negative empirical average difference.  This result empirically suggests that Algorithm~\ref{alg:SMB2} ensures approximate SLA satisfaction, even for small time windows. For instance, user $3$ occasionally has a high difference (46.4 units), mostly due to times $t$ where Algorithm~\ref{alg:SMB2} allocates the user a small amount of resource but a huge load is incoming during the window $[t,t+\tau)$. The experiment tells us that averaging out these ``bad'' times ensures good performance under the SLA criterion. Furthermore, we tested values of $\tau=60$, $500$ and $1000$ minutes; larger windows improve the results, with lower maximum and average values.
\begin{table}
	\caption{Statistics for the difference between the cumulative works of our algorithm and static over time windows $[t,t+\tau)$. For each user $i$ we present the min, max, average, and standard deviation over the sequence of differences $\{ r_i(t) \}_{t}$.}
	\label{tab:MS_diff_SLA}
	
	\centering
	\begin{tabular}{|l|l|l|l|l|}
		\hline
		User   & min    & max  & mean  & std  \\
		\hline
		User 1 & -31.1  & 21.7 & -4.6  & 13.8 \\
		User 2 & -122.3 & 46.9 & -31.2 & 49.2 \\
		User 3 & -90.8  & 46.4 & -7.2  & 29.5 \\
		User 4 & -49.7  & 18.5 & -0.8  & 12.9 \\
		User 5 & -42.6  & 22.6 & -5.85 & 14.0 \\
		User 6 & -21.9  & 14.9 & -0.6  & 6.4 \\
		\hline
	\end{tabular}
\end{table}

\textbf{Queue lengths.} In Figure~\ref{fig:MS_Queues_squared} we present the $2$-norms of queues given by the online benchmark algorithms and Algorithm~\ref{alg:SMB2}. As usual, we can interpret one unit as the respective algorithm's latency, that is, lateness with respect to the overall users' demand. Compared against the online algorithms, we empirically observe the superiority of Algorithm~\ref{alg:SMB2}, as it has the smallest latency most of the time. Algorithm~\ref{alg:SMB2} ends with a $2$-norm of roughly 44 units, average length of 22 and a maximum length of 92. PO ends with a $2$-norm of approximately 68, average length of 32 and a maximum length of 126. OWM ends with a $2$-norm of 68, average of 32 and maximum of 113. Finally, Static shows the worst behavior, with a final $2$-norm of $103$, average of $91$ and maximum of $231$. For this data set, Algorithm~\ref{alg:SMB2} shows a remarkable performance, considering particularly that Algorithm~\ref{alg:SMB2} always reserves $\varepsilon/N$ resource for each user.

\begin{figure}[h!]
	\centering
	\includegraphics[width=0.75\linewidth]{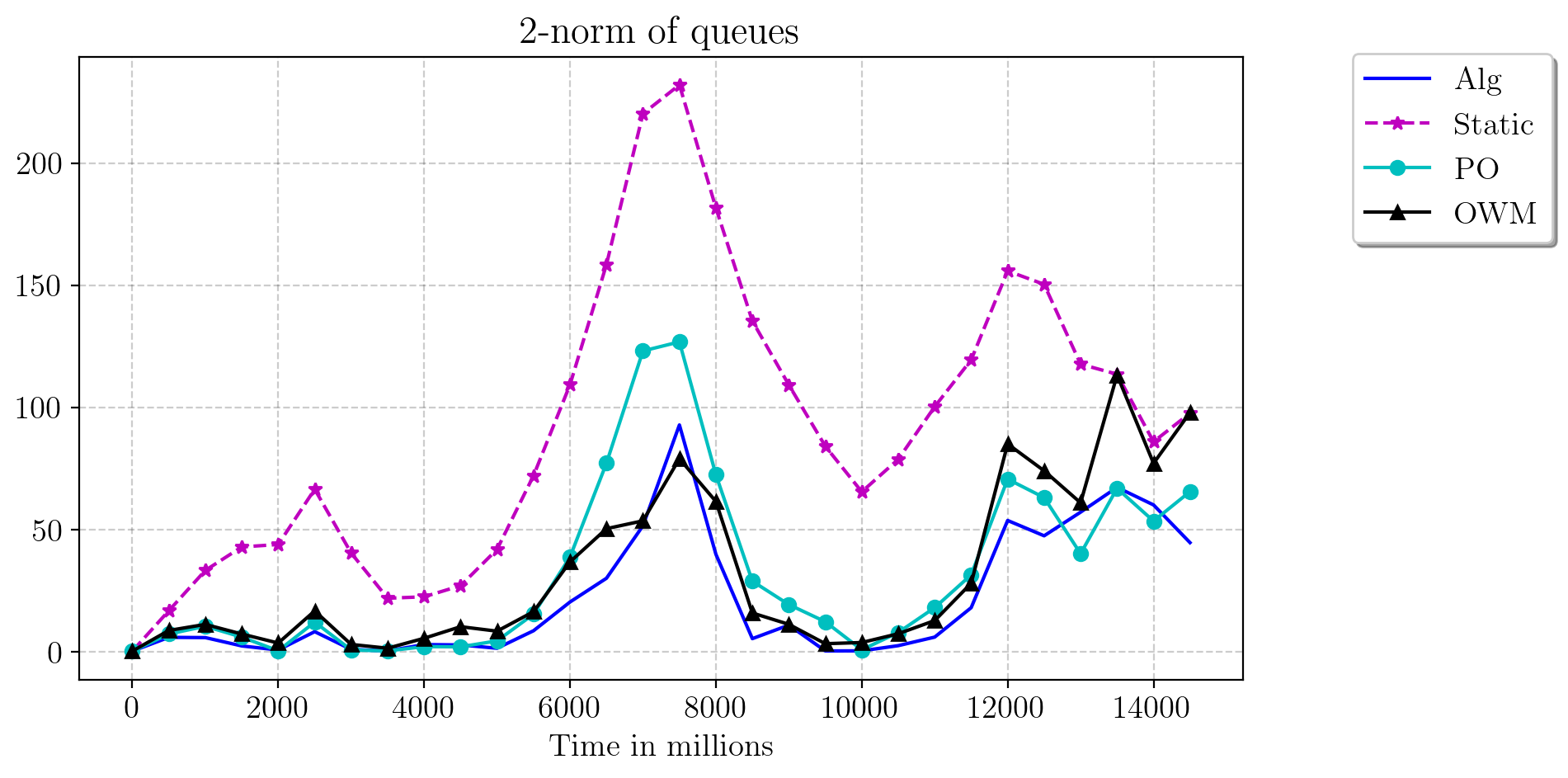}
	\caption{$2$-norm of queues.}
	\label{fig:MS_Queues_squared}
\end{figure}

{\textbf{Summary.} Algorithm \ref{alg:SMB2} performs very well in terms of work maximization compared to all other online algorithms, and exceeds the theoretical guarantees. Furthermore, SLA requirements are typically satisfied, even when measured over relatively short time windows. Finally, as in the previous experiment, the algorithm maintains small queues compared to other online algorithms.
}

\section{Conclusion}
\label{sec:Conclusions}

We have proposed a new online model for dynamic resource allocation of a single divisible resource in a shared system. Our framework captures basic properties of cloud systems, including SLAs, limited system feedback and unpredictable (even adversarial) input sequences. We designed an algorithm that is near-optimal in terms of both work maximization and SLA satisfaction (Theorems~\ref{thm:main}, \ref{thm:mainextension} and \ref{thm:SLAextension}). Furthermore, our second algorithm, Algorithm~\ref{alg:SMB2}, can be applied in an over-commitment regime with similar guarantees, which could be of additional merit for some applications. We derived a simple expression for the offline work maximization problem that allowed us to reinterpret the algorithm's dynamics as an approximate solution of the optimal (offline) work maximization LP. Numerical experiments show that our algorithm is indeed able to achieve a desirable trade-off between work maximization and SLA satisfaction. In particular, comparisons with offline algorithms (PG and restPG) indicate that our algorithm is empirically work maximizing. Further, unlike other plausible online algorithms, our algorithm is able to quickly adapt to unexpected changes in demand and still approximately satisfy the underlying user SLAs.
Our model and results may be extended in various directions of interest to the operations research and cloud computing communities. 

A natural extension for single-resource systems is to model priority among users. Typically, users with higher priority should be given resources before their lower-priority counterparts. A challenge in this setting is how to define the metric corresponding to work maximization. One possibility is to have different weights for different users, corresponding to their priority, and then to maximize the weighted total work while satisfying SLAs. 
Directly extending our algorithms to this case means a user's multiplicative boost depends on their priority. However, our analysis in this paper does not apply, as we use the fact that users' work is interchangeable, whereas the identity of who performs the work is critical in the prioritized case.

Another challenging extension is the management of multiple resources (e.g., CPU, I/O bandwidth, memory), where different users or jobs may require the resources in different proportions. This extension requires a fundamental redefinition of our model, where work done for a user is a function of the multiple resources allocated, and may also depend on a particular job's characteristics. In many real-world scenarios, a job's resource demands are often complementary, e.g.\ RAM and CPU usage. This observation may motivate a possible extension in which we still treat all users' loads as one-dimensional quantities, and the work performed by a user is a relatively simple function of their allocations, e.g.,  a concave non-decreasing function.

\section*{Acknowledgments}

The authors' work was partially supported by the U.S.\ National Science Foundation via grants CMMI 1552479, AF 1910423 and AF 1717947. We thank Vivek Narasayya for useful discussions.
\bibliography{biblio}
\bibliographystyle{ieeetr}

%

%
\appendix
\small

\section{Projecting on $\Delta_\varepsilon$}\label{sec:projection}

\begin{proof}[Proof of Proposition~\ref{prop:projection}]
	Let $\y\in \R_{+}^N$. The projection of $\y$ on $\Delta_\varepsilon$ corresponds to the solution of the convex problem
\[
(Q) \quad \begin{array}{rl}
\min & \sum_i x(i) \ln \left(\frac{x(i)}{y(i)}\right) \\
& \begin{array}{rcll}
\sum_i x(i) & = & 1\\
x(i) & \geq & \varepsilon/N
\end{array}
\end{array}
\]
Its Lagrangian (see \citep{boyd2004convex}) is
\[
\mathcal{L}(\x,\lambda,\mu) = \sum_i x(i) \ln\left( \frac{x(i)}{y(i)}\right) - \lambda\left( \sum_i x(i) -1  \right) - \sum_i \mu_i \left( x(i) -\frac{\varepsilon}{N} \right).
\]
Using the FO conditions:
\[
\forall i : x(i) = y(i) e^{\mu_i +\lambda -1} = y(i) e^{\mu_i} C.
\]
and the SO conditions:
\[
\mu_i \geq 0 ,\,\forall i, \quad \text{ and }\quad x(i) > \frac{\varepsilon}{N} \implies \mu_i =0.
\]
Let $S= \{ i : x(i) = \varepsilon/N  \}$ and $T=[N]\setminus S$. Then, using $\sum_{i} x(i) =1$ we obtain
\[
e^{\lambda-1} = \frac{1- \frac{\varepsilon}{N}|S|}{\sum_{i\in T}y(i)}.
\]
This proves part (b). Now, suppose we have $y(1)\leq \cdots \leq y(N)$. If $i,j\in T$, then $x(i) = y(i) e^{\lambda -1}$ and $x(j) = y(j) e^{\lambda-1}$ and then
\[
x(i) \leq x(j) \iff y(i) \leq y(j).
\]
That is, in $T$ the variables preserve their ordering.

If $i\in S$ and $j\in T$, then $y(i)e^{\lambda-1 +\mu_i} = x(i) = \frac{\varepsilon}{N} < x(j) = y(j) e^{\lambda -1}$, which implies $y(i)<y(j)$ using that $\mu_i\geq 0$. Now, let $k= \min\{ i\in T \}$ which is a well-defined number using constraint $\sum_{i=1}^N x(i)=1$. We claim that for any $ j\geq k$, $j\in T$, that is, $T$ corresponds to the interval $[k,N]$. By contradiction, suppose that $j>k$ does not belong to $T$, then $y(j) < y(k)$ by previous calculus. However $y(j) \geq y(k)$ by the ordering of $\y$. A contradiction. With this, the algorithm to project is clear, we sort $\y$ and then we test increasingly the possible set $S=\{1,\ldots, k-1 \}$ for $k=1,\ldots,N$ and select the best candidate. This proves (a).

We say that $S$ is \emph{feasible} if there is a feasible solution $\x$ such that $S=\{ i : x(i)=\varepsilon/N  \}$. In the following paragraphs we prove that the first feasible solution found in this process is the right one.

Observe that once $S=\{ 1,\ldots,k \}$ is feasible, then $S'=\{ 1,\ldots, j  \}$ remains feasible for all $j\geq k$. Indeed, if $S=\{ 1, \ldots, k \}$ is feasible, then
\[
1 = \frac{\varepsilon}{N}k + \sum_{i\in T} x(i).
\]
Now, increasing $S$ to $S'=\{ 1,\ldots, k+1 \}$ means that we pick $x(k+1)>\frac{\varepsilon}{N}$ and we decrease it to $\frac{\varepsilon}{N}$. Therefore, $x(k+2),\ldots, x(N)$ must increase. Therefore, $S'$ remains feasible. The proof for general case $j\geq k$ follows by induction.

Now, we claim that if $S=\{ 1,\ldots, k \}$ is feasible, then $S'=\{ 1,\ldots, k+1 \}$ cannot have better optimal value. Indeed, the difference between the objective $S'$ and $S$ is
{\small \begin{align*}
	& \sum_{i=1}^{k+1} \frac{\varepsilon}{N}\ln \frac{\varepsilon}{Ny(i)} + \left( 1 - \frac{\varepsilon}{N}(k+1) \right)\ln \frac{1-\frac{\varepsilon}{N}(k+1)}{\sum_{i\geq k+2} y(i) } - \sum_{i=1}^{k} \frac{\varepsilon}{N}\ln \frac{\varepsilon}{Ny(i)} - \left( 1 - \frac{\varepsilon}{N}k \right)\ln \frac{1-\frac{\varepsilon}{N}k}{\sum_{i\geq k+1} y(i) } \\
	=& \frac{\varepsilon}{N}\ln \frac{\varepsilon}{Ny(k+1)} + \left( 1 - \frac{\varepsilon}{N}(k+1) \right)\ln \frac{1-\frac{\varepsilon}{N}(k+1)}{\sum_{i\geq k+2} y(i) } - \left( 1 - \frac{\varepsilon}{N}k \right)\ln \frac{1-\frac{\varepsilon}{N}k}{\sum_{i\geq k+1} y(i) }
	\end{align*}}
The function $f(x)=x\ln x$ is convex for $x>0$. Now, pick $x = \frac{\varepsilon}{N y(k+1)}$, $y = \frac{1-\frac{\varepsilon}{N}(k+1)}{\sum_{i\geq k+2}  y(i)}$ and $\lambda = \frac{y(k+1)}{\sum_{i\geq k+1} y(i)  }$. Then
{\small
	\begin{align*}
	\lambda x + (1-\lambda)y & = \frac{y(k+1)}{\sum_{i\geq k+1} y(i)  }\left( \frac{\varepsilon}{N y(k+1)}\right) + \frac{\sum_{i \geq k+2} y(i) }{\sum_{i\geq k+1} y(i)  } \left(\frac{1-\frac{\varepsilon}{N}(k+1)}{\sum_{i\geq k+2}  y(i)}\right) = \frac{1-\frac{\varepsilon}{N}k}{\sum_{i\geq k+1} y(i).}
	\end{align*}
}
Then, using the convexity of $f$ we obtain the result. This implies that the first feasible prefix $S$ that we find is the optimal one. Therefore, by ordering $\y$ in $\mathcal{O}(N\log N)$ time and then running binary search we can find $S$ in $\mathcal{O}(N\log N)$ time. This finishes the proof of (c).
\end{proof}

\section{Omitted Proofs}\label{sec:Omittedproofs}

\subsection{Proofs of Section~\ref{subsec:LP}}

Here we present dual stated in the offline formulation of the maximum work problem. We have the LP
\[
(P_\varepsilon) \begin{array}{rl}
\max & \sum_{i=1}^N \sum_{t=1}^T w_t(i)\\
& \begin{array}{rcllr}
\sum_{s=1}^t w_s(i) & \leq & \sum_{s=1}^t L_s(i) & \forall t,i  & (1) \\
\sum_{i=1}^{N} w_t(i) & \leq & 1-\varepsilon &\forall t &(2) \\
\w_t & \geq & 0 & \forall t
\end{array}
\end{array}
\]
Using the variables $\alpha_{t}(i)$ for constraint (1) and $\beta_t$ for constraint (2) we obtain the dual
\[
(D_\varepsilon) \begin{array}{rl}
\min & \sum_{i=1}^N \sum_{t=1}^T \alpha_t(i) \sum_{s=1}^t L_s(i) + (1-\varepsilon)\sum_{t=1}^T \beta_t \\
& \begin{array}{rcllr}
\sum_{s=t}^T \alpha_s(i) +\beta_t & \geq & 1 & \forall t,i  & (1') \\
\alpha,\beta & \geq & 0
\end{array}
\end{array}
\]
Using the change of variable $\gamma_t(i) = \sum_{s=t}^T \alpha_s(i)$ we obtain the stated dual
\[
(D_\varepsilon) \begin{array}{rl}
\min & \sum_{i=1}^N \sum_{t=1}^TL_t(i) \gamma_t(i) + (1-\varepsilon)\sum_{t=1}^T \beta_t \\
& \begin{array}{rcllr}
\gamma_t(i) +\beta_t & \geq & 1 & \forall t,i  & (1') \\
\gamma_t & \geq & \gamma_{t+1} & \forall t & (2') \\
\beta, \gamma & \geq & 0
\end{array}
\end{array}
\]
\begin{proof}[Proof of Proposition \ref{prop:optimalLP}]
	
	We prove each inequality separately. Let $0\leq t^\star\leq T$ be such that $\sum_{s=1}^{t^\star} \sum_i L_s(i) + (1-\varepsilon)(T-t^\star)=\min_{0\leq t\leq T} \sum_{s=1}^t \sum_i L_s(i) + (1-\varepsilon)(T-t)$. Consider the dual solution $(\beta,\gamma)$ such that $\gamma_t = 1$, $\beta_t=0$ for $t=1,\ldots,t^\star$ and $\gamma_t = 0,\beta_t=1$ for $t=t^\star+1,\ldots,T$. Then, by weak duality,
	\[
	v_{P_\varepsilon} \leq v_{\text{dual}}(\beta,\gamma) = \sum_{s=1}^{t^\star} \sum_i L_s(i) + (1-\varepsilon)(T-t^\star)=\min_{0\leq t\leq T} \sum_{s=1}^t \sum_i L_s(i) + (1-\varepsilon)(T-t).
	\]
	Now, consider the greedy algorithm that, in each iteration, gives enough allocation to the users in order to complete their work starting with user $1$, then user $2$, and so on. We restrict the algorithms' allocations to $(1-\varepsilon)$-allocations. We denote by $\work_\text{greedy}$ the work done by this algorithm. As usual, we denote by $\w_t$ the vector of work done at time $t$.  Let $t^\star$ be the maximum non-negative $t$ such that $\sum_i w_t(i)<1-\varepsilon$. Observe that $\sum_{s=1}^{t^\star} \sum_i w_s(i) = \sum_{s=1}^{t^\star}\sum_i L_s(i)$. Then
	\[
	\min_{0\leq t\leq T} \sum_{s=1}^t \sum_i L_s(i) + (1-\varepsilon)(T-t) \leq \sum_{s=1}^{t^\star}\sum_i L_s(i) + (1-\varepsilon)(T-t^\star) = \work_{\text{greedy}} \leq v_{P_\varepsilon},
	\]
	since $v_{P_\varepsilon}$ is the optimal solution.
\end{proof}

\begin{remark}
	This max-min result shows that the greedy algorithm is optimal for solving $(P_\varepsilon)$ and also shows how to compute the dual variables. Finally, solving $(P_\varepsilon)$ can be done efficiently in $\mathcal{O}(NT)$ by running the greedy algorithm.
\end{remark}

\subsection{Proofs of Section~\ref{sec:throughput}}

In what follows, we denote by $S_t$ the users with allocation $\frac{\varepsilon}{N}$ at time $t$.

\begin{proof}[Proof of Lemma \ref{lem:monotonicity}]
	\begin{enumerate}
		\item First, for $i\in A_t$ we have
		\[
		h_{t+1}(i) = \frac{h_t(i) e^{\eta g_t(i)} e^{\mu_i} \left( 1- \frac{\varepsilon}{N} |S_{t+1}| \right)  }{\sum_{j\in \overline{S}_{t+1} }\widehat{h}_{t+1}(j) } \geq  h_t(i) \frac{\left( 1- \frac{\varepsilon}{N}|S_{t+1}| \right)}{e^{-\eta}\sum_{j\in \overline{S}_{t+1} }\widehat{h}_{t+1}(j) }.
		\]
		We divide the analysis into two cases: $B_t\cap \overline{S}_{t+1}\neq \emptyset$ and $B_t \subseteq S_{t+1}$.
		
		For $B_t\cap \overline{S}_{t+1}\neq \emptyset$ we have
		\begin{align*}
		e^{-\eta} \sum_{j\in \overline{S}_{t+1}} \widehat{h}_{t+1}(j) & \leq e^{\lambda \eta} \sum_{j\in \overline{S}_{t+1}\cap A_t} h_t(j) + e^{-\eta} \sum_{j\in \overline{S}_{t+1} \cap B_t } h_t(j) \\
		& = e^{ \lambda\eta} \sum_{j\in \overline{S}_{t+1}} h_{t}(j) - (e^{\lambda \eta} - e^{-\eta} ) \sum_{j\in B_t \cap \overline{S}_{t+1}} h_t(j) \\
		& \leq e^{\lambda \eta} \left( 1 -\frac{\varepsilon}{N}|S_{t+1}|  \right) - \frac{\varepsilon}{N}(e^{\lambda \eta} - e^{-\eta} ) \tag{since $h_{t}(i) \geq \frac{\varepsilon}{N}$}\\
		& \leq (1 + 2\lambda \eta) \left( 1 -\frac{\varepsilon}{N}|S_{t+1}|  \right) - \frac{\varepsilon}{N} \left( \lambda\eta + \eta - \frac{\eta^2}{2} \right)  \tag{using $1+ \lambda \eta \leq e^{\lambda \eta}\leq 1+2\lambda \eta$, $e^{-\eta} \leq 1-\eta +\frac{\eta^2}{2}$}\\
		& \leq 1 - \frac{\varepsilon}{N}|S_{t+1}|  + 2\lambda \eta - \frac{\varepsilon}{N} \eta \left( 1 - \frac{\eta}{2}  \right)\\
		& \leq 1 - \frac{\varepsilon}{N}|S_{t+1}| - \frac{\varepsilon \eta}{4N}. \tag{$\eta \leq 1$ and $2\lambda  \leq \frac{\varepsilon}{4N}$}
		\end{align*}
		Therefore,
		\[
		\frac{\left( 1- \frac{\varepsilon}{N}|S_{t+1}| \right)}{e^{-\eta}\sum_{j\in \overline{S}_{t+1} }\widehat{h}_{t+1}(j) } \geq \frac{ 1 - \frac{\varepsilon}{N}|S_{t+1}| }{ 1 - \frac{\varepsilon}{N}|S_{t+1}| - \frac{\varepsilon \eta}{4N} } \geq \frac{1}{1-\frac{\varepsilon \eta}{4N}} \geq 1+ \frac{\varepsilon \eta}{4N},
		\]
		using $\frac{1}{1-x}\geq 1+x$ when $x\in (0,1)$. Hence $h_{t+1}(i) \geq h_{t}(i) (1+\frac{\varepsilon \eta}{4N})$.
		
		Now, if $B_t \subseteq S_{t+1}$, then $\overline{S}_{t+1} \subseteq A_{t}$. We have
		\begin{align*}
		\frac{e^{-\eta } \sum_{j\in \overline{S}_{t+1}} \widehat{h}_{t+1}(j)}{1- \frac{\varepsilon}{N}|S_{t+1}|} & \leq \frac{(1-\varepsilon)e^{\lambda \eta}}{1-\frac{\varepsilon}{N}|S_{t+1}|} \\
		& \leq \frac{(1-\varepsilon)e^{\lambda \eta}}{1-\varepsilon+ \frac{\varepsilon}{N}}\\
		& \leq \frac{(1-\varepsilon)(1+2\lambda \eta)}{1-\varepsilon + \frac{\varepsilon}{N}} \tag{$e^{\lambda \eta} \leq 1+ 2\lambda \eta$ since $2\lambda \eta < 1$}\\
		& \leq \frac{1-\varepsilon + 3\lambda \eta}{1-\varepsilon + \frac{\varepsilon}{N}} = 1 - \frac{\frac{\varepsilon}{N} + 3\lambda \eta }{1-\varepsilon + \frac{\varepsilon}{N}}.
		\end{align*}
		Since $\frac{\frac{\varepsilon}{N} + 3\lambda \eta }{1-\varepsilon + \frac{\varepsilon}{N}} \geq \frac{\varepsilon}{N}$ we obtain
		\[
		h_{t+1}(i) \geq h_t(i) \frac{1}{1- \frac{\varepsilon}{N}} \geq h_{t}(i)\left( 1+\frac{\varepsilon}{N}  \right) \geq h_t(i) \left( 1+ \frac{\varepsilon \eta}{4N} \right).
		\]		
		
		\item The monotonicity of $h_t(i)$ with $i\in A_t^1$ is easy to see. Let us prove the second statement:
		\begin{align*}
		\frac{e^{-\eta} \sum_{j\in \overline{S}_{t+1}} \widehat{h}_{t+1}(j)}{1-\frac{\varepsilon}{N}|S_{t+1}|} & \leq \frac{ e^{\lambda \eta} \sum_{j \in \overline{S}_{t+1}} h_t(j)}{1-\frac{\varepsilon}{N}|S_{t+1}|}\\
		& \leq \frac{e^{\lambda \eta} \left( 1- \frac{\varepsilon}{N}|S_{t+1}|  \right)}{1-\frac{\varepsilon}{N}|S_{t+1}|}\\
		& \leq e^{\lambda \eta}\\
		& \leq 1 + 2\lambda \eta = 1 + \varepsilon c,
		\end{align*}
		since $\lambda = \frac{\varepsilon^2}{8N}$. Then, for $i\in A_t^2$,
		\[
		h_{t+1}(i) \geq h_t(i) \frac{e^{\eta} (1-\frac{\varepsilon}{N}|S_{t+1}|)}{\sum_{j\in \overline{S}_{t+1}} \widehat{h}_{t+1}(j) } \geq h_t(i) \frac{1}{1+\varepsilon c} \geq h_t(i)(1-\varepsilon c).
		\]
	\end{enumerate}

\end{proof}


\begin{proof}[Proof of Claim \ref{cl:gooduser2}]
	Now, let $[s^\star+1, \ldots, T]$ and let us divide this interval into blocks of length $\widetilde{s}$ with a possible last piece of length of length at most $\widetilde{s}$. Let $L$ be one of these blocks and let $i$ be the user given by claim \ref{cl:gooduser1}, that is, $M_{i,r}=1$ for all $r\in L$. Consider $\textstyle L' = \{  t \in L :  \sum_{j\in A_t} h_t(j) < 1- \varepsilon  \}$. By using part 1 of Lemma~$\ref{lem:monotonicity}$, user $i$ increases her allocation multiplicatively in $L'$ by a factor of $(1+c)$. Observe that for $t\notin L'$, user's $i$ allocation can increase or decrease depending on $h_t(i)$. However, by Lemma by part 2 of \ref{lem:monotonicity}, we know that $h_t(i)$ will not decrease by a huge amount. Let $k'=|L'|$, then $i$ increases her allocation for $k'$ times and decreases it for at most $\widetilde{s}-k'$ times. Therefore, $k'$ maximum value is such that
	\[
	\frac{\varepsilon}{N}(1+c)^{k'} (1- \varepsilon c)^{\widetilde{s}-k'} =1
	\]
	and therefore,
	\begin{align*}
	k' & \leq \frac{\ln (N/\varepsilon) +\widetilde{s} \ln(1-\varepsilon c)^{-1}}{\ln((1+c)/(1-\varepsilon c))}\\
	& \leq \frac{1+c}{c(1+\varepsilon)} \ln (N/\varepsilon) + \frac{\varepsilon(1+c)}{(1+\varepsilon)(1-\varepsilon c)} \widetilde{s} \tag{$\ln \frac{1+c}{1-\varepsilon c} \geq \frac{c(1+\varepsilon)}{1+c}$, $\ln \frac{1}{1-\varepsilon c}\leq \frac{\varepsilon c}{1-\varepsilon c}$}\\
	& \leq \frac{\varepsilon(1+c)}{(1+\varepsilon)} \widetilde{s} + \frac{\varepsilon (1+c)}{(1+\varepsilon)(1-\varepsilon c)} \widetilde{s} \tag{$\widetilde{s}=\frac{\ln (N/\varepsilon)}{\varepsilon c}$}\\
	& = \varepsilon \frac{1+c}{1+\varepsilon} \left( 1 +\frac{1}{1-\varepsilon c} \right)\widetilde{s} \\
	& \leq 3 \varepsilon \widetilde{s}. \tag{for $N\geq 2$ and $\varepsilon\leq \frac{1}{10}$}
	\end{align*}
	Hence, $L'$ is at most a fraction of $\widetilde{s}$ and with this
	\[
	\sum_{t\in L} \sum_{i} w_t(i) \geq (1-\varepsilon)(\widetilde{s}-k') \geq  (1-\varepsilon )\left(1- 3\varepsilon   \right) \widetilde{s} \geq (1-4\varepsilon)|L|.
	\]
	Summing over all blocks we conclude the desired result.
\end{proof}

\subsection{Proof of Section~\ref{sec:SLAsat}}

\begin{proof}[Proof of Lemma \ref{lem:monotonicity2}]
	If $A_t^1=\emptyset$, the result is vacuously true. Suppose that $A_t^1\neq \emptyset$. First, we prove that under the assumption of Lemma~\ref{lem:monotonicity2}, we have $\overline{A}_{t}^1 \cap \overline{S}_{t+1} \neq \emptyset$. For $j\in S_{t+1}$ we have
	\begin{align*}
	\frac{\varepsilon}{N} &= h_{t+1}(j) \\
	&=\frac{\widehat{h}_{t+1}(j) e^{\mu_i} (1-\frac{\varepsilon}{N}|S_{t+1}|) }{\sum_{k\in \overline{S}_{t+1} } \widehat{h}_{t+1}(k) }\\
	& \geq \frac{h_{t}(j) (1-\frac{\varepsilon}{N}|S_{t+1}|) }{ \sum_{k\in \overline{S}_{t+1}} h_t(k) e^{\eta(1+\lambda)}  } \tag{since $\mu_i \geq 0$},
	\end{align*}
	This implies $h_{t}(j) \leq \frac{\varepsilon}{N}e^{\eta(1+\lambda)} < \beta(j)$. Since  $\sum_{j} h_t(j) =1 $, $i\in A_t^1$ and $\sum_{j} \beta(j)=1$ we must have that there is a user $j\neq i$ with allocation $h_t(j)\geq \beta(j)$. Clearly, $j\notin S_{t+1}$ and $j\notin A_t^1$. Therefore, $\overline{A}_t^1 \cap \overline{S}_{t+1}\neq \emptyset$.
	
	Following the proof of Lemma~\ref{lem:monotonicity}, for $i\in A_t^1$, we have
	\begin{align*}
	{e^{-\eta(1+\lambda)}\sum_{j\in \overline{S}_{t+1}} \widehat{h}_{t+1}(j)  } & \leq {\sum_{j\in \overline{S}_{t+1}\cap A_{t}^1} h_{t}(j)   + e^{-\eta\lambda} \sum_{j\in \overline{S}_{t+1} \cap \overline{A}_t^1 } h_{t}(j) }\\
	& = \sum_{j\in \overline{S}_{t+1}} h_t(j) - (1-e^{-\lambda \eta}) \sum_{j\in \overline{S}_{t+1}\cap \overline{A}_t^1 } h_t(j) \\
	& \leq 1-\frac{\varepsilon}{N}|S_{t+1}| - \frac{\varepsilon}{N}(1-e^{-\lambda\eta}) \tag{since $\overline{S}_{t+1}\cap \overline{A}_t^1 \neq \emptyset$}\\
	& \leq 1- \frac{\varepsilon}{N}|S_{t+1}| - \frac{\varepsilon \lambda \eta}{2N}. \tag{$1-e^{-x}\geq \frac{x}{2}$ for $x\in [0,1]$}
	\end{align*}
	Therefore,
	\[
	h_{t+1}(i) \geq h_{t}(i) \frac{1-\frac{\varepsilon}{N}|S_{t+1}|}{1-\frac{\varepsilon}{N}|S_{t+1}| - \frac{\varepsilon\lambda\eta}{2N} } \geq h_{t}(i) \left( 1+ \frac{\varepsilon\lambda\eta}{2N} \right).
	\]
\end{proof}

\section{Greedy Online Algorithm}\label{sec:localgreedy}

In this section, we prove that the following greedy allocation strategy is almost optimal in work maximization. The algorithm divides the users into 3 categories: $A$ non-empty queue users with non-zero allocation, $B$ non-empty queue users with zero allocation and $I$ empty queue users with zero allocation. At time $t$, a user $i\in A$ is left in $A$ if she still has non-empty queue, otherwise we will move her to $I$; a user $i\in I$ will be moved to $B$ if her queue is becomes non-empty, otherwise she will remain in $I$; finally, if all users from $A$ are moved to $I$, then we will move all $B$ to $A$, otherwise we will left $B$ untouched. In any case, we will distribute uniformly among the users that remain in $A$.

Users move from $A$ to $I$, $I$ to $B$ and $B$ to $A$. Let $\w_t$ be the work done by the algorithm and let $\w_t'$ be the optimal offline work.

\begin{theorem}
	For any loads $L_1,\ldots,L_T\in \R_{\geq 0}^N$, and any $\varepsilon>0$, this greedy Algorithm guarantees
	\[
	\sum_{t=1}^T \sum_i w_t(i) + 2\frac{N^2}{\varepsilon} \geq \sum_{t=1}^T \sum_i w_t'(i)
	\]
	where $\w_t'$ is the work done by the optimal offline sequence of $(1-2\varepsilon/N)$-allocations.
\end{theorem}

\begin{proof}
	Let $t^\star$ be the maximum $t\geq 0$ such that $\sum_{s=1}^{t+N^2} \sum_{i} w_s(i)\geq \sum_{s=1}^{t}\sum_i L_s(i)$. By claim \ref{cl:gooduser1} we know that each interval $[r,r+N^2/\varepsilon)$, with $r>t^\star$, has a user with no-empty queue. As in claim \ref{cl:gooduser2} we divide the interval $[t^\star+1,T]$ into blocks of length $N^2/\varepsilon$ with a last block of length at most $N^2/\varepsilon$. Pick any of these blocks, say $L$, and let $L'=\{ t\in L : \sum_i w_t(i)<1 \}$. It is easy to see that $|L'|\leq 2 N$ and therefore, summing over all block, we have $\sum_{t\geq t^\star+1}\sum_{i} w_t(i)+N^2/\varepsilon\geq (T-t^\star)(1-2\varepsilon/N)$. The conclusion follows applying weak duality to $(P_{2\varepsilon/N})$.
\end{proof}

\begin{remark}
	Against the best $1$-allocations we can optimize $\varepsilon$ and obtain $\varepsilon = \sqrt{N^3/T}$. This greedy strategy will be $\mathcal{O}(\sqrt{N T})$ far from the optimal dynamic work. Observe that this matches the lower bound in Theorem~\ref{thm:GeneralLB}.
\end{remark}

\section{Lower Bound}\label{sec:Lower}

\begin{theorem}\label{thm:GeneralLB}
	For any online deterministic algorithm $\mathcal{A}$ setting at each time $1$-allocations, with an underlying queuing system, and with the same limited feedback as Algorithm~\ref{alg:SMB1}, there exists a sequence of online loads $L_1,\ldots,L_T$ such that $\work_{\h_1^*,\ldots,\h_T^*} - \work_\mathcal{A} = \Omega\left(\sqrt{T}\right)$, where $\h_1^*,\ldots,\h_T^*$ are the optimal offline dynamic $1$-allocations.
\end{theorem}

\begin{proof}
	We consider the case with $N=2$ users, the general case reduces to $N=2$ by loading jobs only to two users. Let $\mathcal{A}$ be an online algorithm for allocating a divisible resource for $2$ users and with underlying queuing system and limited feedback. Without loss of generality, we can assume that the allocations sets by $\mathcal{A}$ always add up to $1$ at every time step.

	We will construct a sequence of loads $L_t = (L_t(1), L_t(2))$ that at every time will add up to $1$. This will ensure that the overall work done by the optimal offline dynamic policy will be $T$. On the other hand, we will show that this sequence of loads will lead to large queue length for at least one of the users.  The main ingredient is to use the fact the algorithm receives limited feedback about the state of the system, i.e., which users have empty queue. In particular, this implies that if there are two distinct set of load vectors $L_t$ and $L'_t$ for  some interval $t\in [r,s]$ such that the queues remain non-empty on both these sequences, then the resource allocation to the users in the two load sequences must be identical. 
	
	We will divide the time window $[1,T]$ into phases. Each phase will begin with a configuration of queues, say $Q= (Q(1), Q(2) )$, where one of the queues is empty and the other one nonempty. We set $q= Q(1)+Q(2)$ and we denote by $q_i$ the $q$ at phase $i$. We define $q_0 = 0$. We will prove that at the end of each phase $i\geq 1$, $q_{i+1}\geq q_i+\frac{1}{4}$ with all $q_{i+1}$ cumulated in one queue and the other queue empty.
	
	Initially, the algorithm has a fixed deterministic allocation $\h_1 = (h_1(1), h_1(2) )$. If $h_1(1)\leq h_1(2)$, then we load $L_1 =(0,1)$. Otherwise, we load $L_1 = (1,0)$. In any case, we have $q_1\geq \frac{1}{2}$ and all $q_1$ in one queue.
	
	Now, we will describe how the general phases work. For the sake of simplicity, we will describe the phase starting at time $t=1$. We have queue configuration $Q  =  ( Q(1), Q(2) ) $ with $q>0$. By the initial phase, we can assume $q\geq 1/2$. Moreover, we can assume that only one of the queues is nonempty, this point will be clear after we describe how the phase works and it is clearly true for phase $1$. Phase $i$ with $q=q_i$ will last at most $2q + 2$ time steps.
	
	Suppose that $Q(1)=0$ and $Q(2)>0$. If $h_1(1)=1$, then we load $L_1=(0,1)$ and the phase ends with $q$ increased by $1$ and user $1$'s queue empty. Therefore, we can assume $h_1(1)<1$. Our first load will be $L_1 =  (h_1(1)+\varepsilon, h_2(2)-\varepsilon) $ with $0<\varepsilon<\frac{1}{4}$ small enough and such that both queues are nonempty. The following loads will be $L_t=\h_t$, the allocation of $\mathcal{A}$ at time $t$. Observe that the first load will ensure that both users see nonempty queues until the end of the phase. Moreover, user $1$ always has exactly $\varepsilon$ remaining in her queue.
	
	\begin{itemize}
		\item If there is a time $\tau^\star \in [1, 2q+1]$ such that $h_{\tau^\star}(1)\geq 1/2$, then we change the load at time $\tau^\star$ for $L_{\tau^\star}' = (0,1)$. This will increment $q$ by at least $1/2-\varepsilon \geq 1/4$ and the phase ends. Observe that $Q_{\tau^\star+1}(1)=0$.
		
		\item We can assume now that for the loads $L_t=\h_t$ we always have $h_t(1)<1/2$ for all $t\in [2,2q+1]$. We change the loads to $L_t'=(1,0)$ until time $\tau^\star$ in which user $2$ empties her queue. Recall that the feedback of the algorithm is only the set of empty queues at every time step. Thus the behavior of $\mathcal{A}$ under $L_t$ and $L_t'$ will be the same until time $\tau^\star$. Now, we change load $L_{\tau^\star}$ by $L_{\tau^\star}' = ( 1 - h_{\tau^\star}(2) + Q_{\tau^\star}(2) -\varepsilon'  ,  h_{\tau^\star}(2) - Q_{\tau^\star}(2) + \varepsilon'  )$ with $0<\varepsilon'<1/4$ small enough. This will ensure that queue $2$ will be exactly $\varepsilon'$. Now, in an extra step, we load $L_{\tau^\star+1}=(1,0)$. Again, we have $q$ increased by at least $1/4$ and this ends the phase. Observe that $Q_{\tau^\star+2}(2)=0$.
	\end{itemize}

	The analysis is similar for $Q(1)>0$ and $Q(2)=0$. Observe that at the end of each phase, only one queue is nonempty and the other one is empty. In any case, we have the desired increment. With this, we can set the following recurrence, $
	q_0 = 0, \quad q_{i-1} + \frac{1}{4} \leq q_{i} \quad \forall i\geq 1$.	We deduce that $q_i \geq i/4$. Now, let $m$ be the number of phases. By construction, each phase last at most $2 q_i +2$. Then $T \leq \sum_{i=1}^m (2q_i +2) \leq 40 q_m^2$, where we have used $ 4q_m\geq4 q_i \geq i \geq 1$. From here we deduce that $q_m \geq \sqrt{T/40}$.
	
	Now, the work done by the algorithm and the unfulfilled work in the queues must add up the overall load. Then $q_m + \work_\mathcal{A} = T = \work_{\h_1^*,\ldots,\h_T^*}$	from which we obtain the result.
\end{proof} 

\section{Additional Algorithms}\label{sec:Add_Alg}

\begin{center}
	\begin{algorithm}[H]
		\KwIn{Sequence of loads $(L_t)_{t=1}^T$ and SLAs $\beta(1),\ldots,\beta(N)$.}
		\For{$t=1,\ldots,T$}
		{
			$w_t(i) \leftarrow 0$, $\forall i\in [N]$. \\
			$\mathrm{Rem} \leftarrow 1$.\\
			$\mathrm{Rem}(i) \leftarrow Q_{t-1}(i) + L_t(i)$, $\forall i \in [N]$. \\
			\Repeat{$\mathrm{Rem}=0$ or $\{ i : \mathrm{Rem}(i)> 0 \}= \emptyset$}{
				$A \leftarrow \{  i : \mathrm{Rem}(i) >0  \}$.\\
				$\Sigma \leftarrow  \sum_{i\in A_t} \beta(i)$.\\
				$i^*\leftarrow \argmin_{k\in A} \mathrm{Rem}(k) $.\\
				\If{$\mathrm{Rem}(i)< \frac{\beta(i)}{\Sigma} \mathrm{Rem} $}{
					$w_t(i^*)\leftarrow w_t(i^*)+ \mathrm{Rem}(i^*)$.\\
					$\mathrm{Rem} \leftarrow \mathrm{Rem} - \mathrm{Rem}(i^*)$.\\
					$\mathrm{Rem}(i^*)\leftarrow 0$.
				}
				\Else{
					\For{$k\in A$}{
						$w_t(k) \leftarrow w_t(k) + \frac{\beta(k)}{\Sigma}  \mathrm{Rem}$.\\
						$\mathrm{Rem}(k) \leftarrow \mathrm{Rem}(k) - \frac{\beta(k)}{\Sigma}  \mathrm{Rem}$.
					}
					$\mathrm{Rem} \leftarrow 0.$
				}
			}
		}
		\caption{Proportional Greedy} \label{alg:PropGreedy}
	\end{algorithm}
\end{center}

\begin{center}

	\begin{algorithm}[H]
		\KwIn{Sequence of loads $(L_t)_{t=1}^T$ and SLAs $\beta(1),\ldots,\beta(N)$.}
		Initial distribution $\h_1= (\beta(1),\ldots, \beta(N))$,\\
		\For{$t=1,\ldots,T$}
		{
			Set $\h_t$ and obtain $A_t=\{ i : Q_t(i)\neq 0  \}$,\\
			Update $h_{t+1}(i) = \frac{\beta(i)}{\sum_{j\in A}\beta(j)}$\\
			If $A=\emptyset$, then $\h_{t+1}= \h_1$.
		}
		\caption{Online Proportional} \label{alg:OnlProp}
	\end{algorithm}
\end{center}

\section{Gamma Distribution}\label{sec:Gamma}

Recall that a Gamma distribution (\cite{feller1957introduction}) is characterized by two parameters: the \emph{shape} $k> 0$ and the \emph{scale} $\theta> 0$. The PDF of a Gamma$(k,\theta)$ is given by $\frac{1}{\Gamma(k) \theta^k} x^{k-1} e^{-x/\theta} $ where $\Gamma(k)= \int_{0}^{\infty} u^{k-1} e^{-u} \mathrm{d}u$ is the standard Gamma function. See Figure~\ref{fig:gammaplot} for PDFs of different Gamma distribution for various choices of $k$ and $\theta$.

\begin{figure}[h!]
	\centering
	\includegraphics[width=0.75\linewidth]{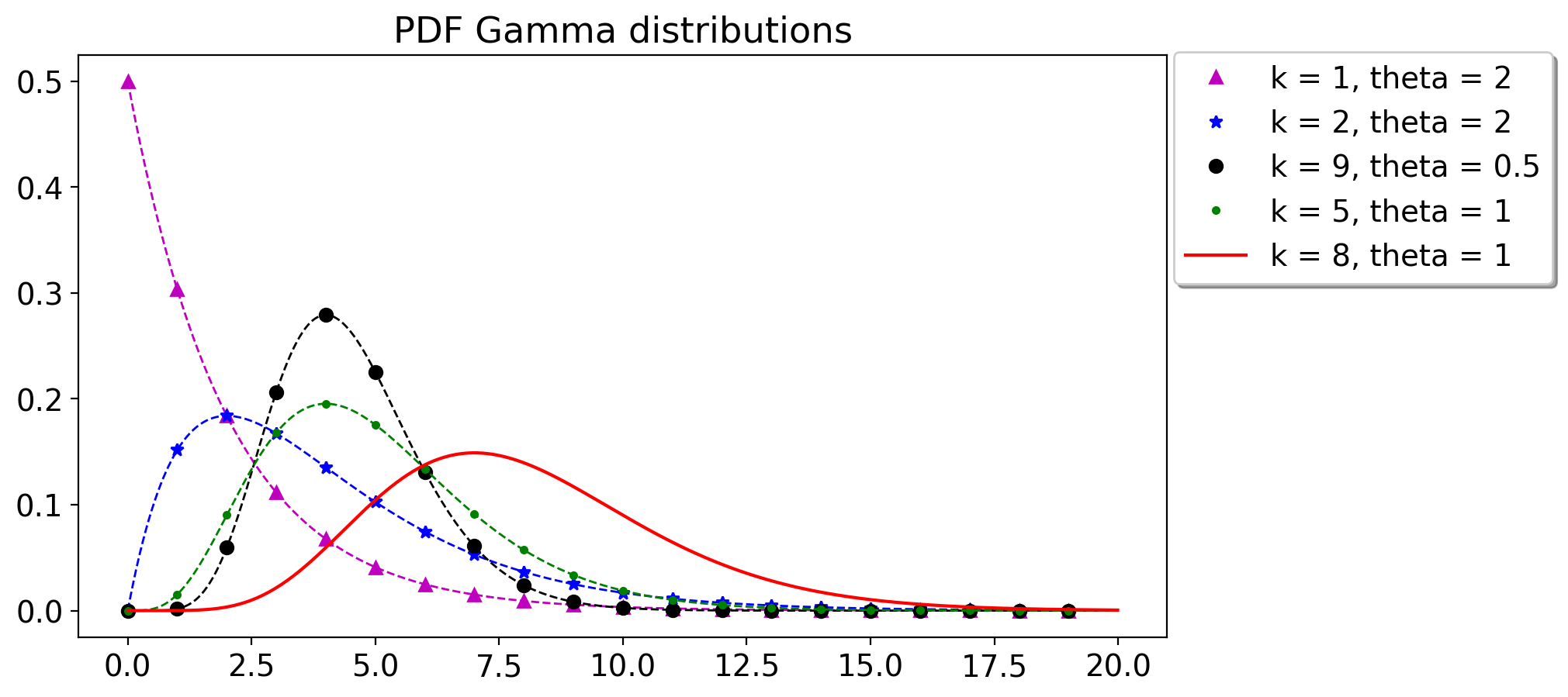}
	\caption{PDF of different Gamma distributions.}
	\label{fig:gammaplot}
\end{figure}

\begin{proposition}
	Let $X\sim \mathrm{Gamma}(k,\theta)$. Then, $\E[X] = k\theta$ and $\mathrm{Var}(X) = k\theta^2$.
\end{proposition}

\end{document}